\documentclass[preprint,12pt]{elsarticle}

\usepackage{amssymb}
\usepackage{mathtools}
\usepackage{amsfonts}
\usepackage{todonotes}
\usepackage{macros}
\usepackage{complexity}
\usepackage{hyperref}
\usepackage{stmaryrd}

\journal{}

\newtheorem{theorem}{Theorem}
\newtheorem{definition}[theorem]{Definition}
\newtheorem{lemma}[theorem]{Lemma}
\newtheorem{proposition}[theorem]{Proposition}
\newtheorem{corollary}[theorem]{Corollary}
\newtheorem{rmk}[theorem]{Remark}
\newtheorem{claim}[theorem]{Claim}
\newproof{proof}{Proof}
\newproof{pot}{Proof of Theorem \ref{thm2}}

\begin{document}

\begin{frontmatter}



\title{Context-Free Commutative Grammars with Integer Counters and
  Resets}


\author[sws]{Dmitry Chistikov\fnref{fn2}}
\address[sws]{Max Planck Institute for Software Systems (MPI-SWS), Germany}

\author[cachan]{Christoph Haase\fnref{fn1}}
\address[cachan]{LSV, CNRS \&
  ENS Cachan, Universit\'e Paris-Saclay, France}
\fntext[fn2]{Sponsored in part by the ERC Synergy award ImPACT.
  Present address: Department of Computer Science, University of Oxford, UK.}
\fntext[fn1]{Supported by Labex Digicosme, Univ. Paris-Saclay, project
  VERICONISS.}

\author[cachan]{Simon Halfon}

\begin{abstract}
  We study the computational complexity of reachability, coverability
  and inclusion for extensions of context-free commutative grammars
  with integer counters and reset operations on them. Those grammars
  can alternatively be viewed as an extension of communication-free
  Petri nets. Our main results are that reachability and
  coverability are inter-reducible and both NP-complete. In particular,
  this class of commutative grammars enjoys semi-linear reachability
  sets. We also show that the inclusion problem is, in general,
  coNEXP-complete and already $\Pi_2^\text{P}$-complete for grammars
  with only one non-terminal symbol. Showing the lower bound for the
  latter result requires us to develop a novel
  $\Pi_2^\text{P}$-complete variant of the classic subset sum
  problem.
\end{abstract}


\end{frontmatter}


\section{Introduction}

This paper studies the computational complexity of certain decision
problems for extensions of context-free commutative grammars with
integer counters and reset operations on them. The motivation for our
work comes from the close relationship of such grammars with
subclasses of Petri nets. For presentational purposes, we begin with
introducing the decision problems we consider in terms of
Petri nets.

Petri nets, or equivalently Vector Addition Systems with States
(\VASS), are a prominent and appealing class of infinite-state
systems, from both theoretical and practical perspectives. On
the one hand, their high level of abstraction allows them to be used
as a mathematical model with well-defined semantics in a wide range of
application domains, in particular but not limited to the verification
of concurrent programs, see, e.g.,~\cite{GS92}. On the other hand, for
half a century Petri nets have provided a pool of challenging and
intricate decision problems and questions about their structural
properties. One of the most important and well-known instance is the
question about the computational complexity of the reachability
problem for Petri nets, which has attracted the attention of
generations of researchers without, however, having been fully
resolved.

A Petri net comprises a finite set of places with a finite number of
transitions. Places may contain a finite number of tokens, and a
transition can consume tokens from places, provided sufficiently many
are present, and then add a finite number of tokens to some places. In
the \VASS\ setting, places are referred to as counters and we will
often use these terms interchangeably in this paper. A
configuration of a Petri net is a marking of its places, which is
just a function $\vec{m} \colon \mathrm{Places}\to \nat$ or, equivalently,
a vector of natural numbers whose components are indexed by elements
from $\mathrm{Places}$. The most prominent decision problems for Petri
nets are reachability, coverability and inclusion. Given
configurations $\vec{m}$ and $\vec{n}$ of a Petri net $\mathcal{A}$,
reachability is to decide whether there is a sequence of transitions
of $\mathcal{A}$ whose effect transforms $\vec{m}$ into
$\vec{n}$. Coverability asks whether there is a transition sequence
from $\vec{m}$ to a configuration that is ``above'' $\vec{n}$, i.e., a
path to some configuration $\vec{n}'$ such that $\vec{n}'\ge \vec{n}$,
where $\ge$ is interpreted component-wise. Finally, given Petri nets
$\mathcal{A}$ and $\mathcal{B}$ with the same set of places, inclusion
asks whether the set of markings reachable in $\mathcal{A}$ is
contained in the set of those reachable in $\mathcal{B}$. All of these
problems have been extensively studied in the
literature. One of the earliest results was obtained by Lipton, who
showed that reachability and coverability
are \EXPSPACE-hard~\cite{Lipt76}. Subsequently, Rackoff established a
matching upper bound for coverability~\cite{Rack78}, and Mayr showed
that reachability is decidable~\cite{Mayr81}. This result was later
refined~\cite{Kos82,Lam92} and shown in a different way in~\cite{Leroux12}, and an actual
complexity-theoretic upper bound, namely membership in
$\mathbf{F}_{\omega^3}$, a level of the fast-growing hierarchy, was
only recently established~\cite{LS15}. For inclusion, it is known that
this problem is in general undecidable~\cite{Hack76} and Ackermann
($\mathbf{F}_{\omega}$)-complete when restricting to Petri nets with a
finite reachability set~\cite{Jan01}.

For some application domains, standard Petri nets are not
sufficiently expressive. For instance, as discussed, e.g.,
in~\cite{KKW14}, the verification of concurrent finite-state
shared-memory programs requires additional operations on places such
as transfers, where the content of one place can be copied to another
one. Another example is the validation of business processes, which
requires reset operations on places, i.e., a special kind of
transitions which assign the value zero to some
place~\cite{WAHE09}. The computational price for these
extensions is high: reachability in the presence of any such extension
becomes undecidable~\cite{DFS98,FGH13}, while the complexity of
coverability increases significantly to Ackermann
($\mathbf{F}_{\omega}$)-completeness in the presence of
resets~\cite{Schn10}.

One of the main sources of the high complexity of decision problems
for Petri nets and their extensions is the restriction that the places
contain a \emph{non-negative} number of tokens. This restriction
enables one to enforce an order in which transitions can be taken, which
is at the heart of many hardness proofs. In this paper, we relax this
restriction and study the computational complexity of decision
problems for a subclass of Petri nets, where the nets have additional counters that
range over the integers and can be reset and where transitions
are also structurally restricted. One advantage of this class is the
decidability and a much lower computational complexity of standard
decision problems when compared to usual Petri nets with reset
operations.

\subsubsection*{Our contribution.} 

The main focus of this paper is the computational complexity
of reachability, coverability and inclusion for so-called
communication-free Petri nets extended with integer counters and
resets, and for subclasses thereof. A communication-free Petri net is a
Petri net in which every transition can remove a token from at most
one place. An important property of communication-free Petri nets is
that their sets of reachable markings are
semi-linear~\cite{Huy85,Yen97,MW15}, meaning in particular that they
are closed under all Boolean operations (this is not the case for
general Petri nets~\cite{HP79}). Communication-free Petri nets are
essentially equivalent\footnote{This will be made more precise in
  Section~\ref{ssec:relationship}.}  to context-free commutative
grammars, or basic parallel processes, and have extensively been studied
in the
literature~\cite{Huy83,Huy86,Esp97,Yen97,KT10,MW15,Kop15,HH15}. For
technical convenience we adopt the view of communication-free Petri
nets as context-free commutative grammars in the technical part of
this paper.

As our first main result, we show that context-free commutative
grammars can be extended by a finite number of integer counters, i.e.,
counters that range over the integers and can be reset by
transitions, while retaining \NP-completeness of reachability and
coverability, as well as preserving semi-linearity of the reachability set. This is
achieved by showing that the reachability set of our extended class
can be defined by a formula in existential Presburger arithmetic of
polynomial size. The characterization obtained in this way can then be
used in order to show \coNEXP-completeness of the inclusion problem by
application of complexity bounds for Presburger arithmetic.

Our second main result is a more refined analysis of the
complexity of the inclusion problem. We show that even in the
structurally simplest case of context-free commutative grammars with
integer counters and \emph{without any} control structure, i.e., a
singleton non-terminal alphabet, the inclusion problem is hard for the
second level of the polynomial hierarchy and, in fact,
$\ComplexityFont{\Pi_2^\P}$-complete. In essence, this problem is
equivalent to asking, given two integer matrices $A$, $B$ and a vector
$\vec{v}\in \nat^d$, whether for all $\vec{x}\in \nat^m$ there exists
some $\vec{y}\in \nat^n$ such that $A \cdot \vec{x} + B \cdot \vec{y}
= \vec{v}$. We prove hardness of this problem by developing a new
$\ComplexityFont{\Pi_2^\P}$-complete variant of the classical
\textsc{Subset Sum} problem, which we believe is a contribution of
independent interest.

This paper is an extended version of our conference paper~\cite{HH14},
which appeared in the proceedings of the 8th International Workshop on
Reachability Problems (RP~2014) held in Oxford, UK, in September 2014. It extends
the results from~\cite{HH14} by considering a more general model:
context-free commutative grammars with integer counters and resets
instead of integer vector addition systems with states and resets
considered in~\cite{HH14}; we also provide full proofs.
Moreover, in~\cite{HH14} we
left as an open question the precise complexity of the aforementioned
$\ComplexityFont{\Pi_2^\P}$-complete inclusion problem, which we could
only show to be \NP-hard and in $\ComplexityFont{\Pi_2^{\P}}$.
This question is now resolved in this paper.

\subsubsection*{Related Work.}
Apart from the related work mentioned above, closely connected to the
problems considered in this paper is the work by Kopczy{\'n}ski and
To~\cite{KT10} and Kopczy{\'n}ski~\cite{Kop15}. In their work, the
complexity of various decision problems for context-free commutative
grammars and subclasses thereof has been studied when the number of
alphabet symbols (which roughly corresponds to the number of places
in the Petri net representation) is fixed. In~\cite{Kop15}, alphabet symbols
may, informally speaking, be erased and negative quantities of
alphabet symbols are possible. This essentially corresponds to adding
counters with integer values to context-free commutative grammars. A
further generalization of communication-free Petri nets,
recently studied by Mayr and Weihmann, are communication-free
Petri nets with arbitrary edge multiplicities~\cite{MW13}. In this
class, transitions may also only consume tokens from one place, but
they may take an arbitrary number of them.

A powerful technical tool in this context that we employ in this paper
is defining Parikh images of communication-free Petri nets in
existential Presburger arithmetic. This approach has been directly or
indirectly taken, for instance,
in~\cite{Esp97,PR99,SSMH04,VSS05,HKOW09,HL11}. In particular, in this
paper we generalize a technique of Verma et al.~\cite{VSS05}, which
has also been done in~\cite{HL11} in order to show decidability and
complexity results for pushdown systems equipped with reversal-bounded
counters.

As discussed above, we achieve a lower complexity for standard
decision problems in comparison to general Petri nets by relaxing
counters to range over the integers. Another approach going into a
similar direction is to allow counters to range over the positive
reals, for instance in continuous Petri nets introduced
in~\cite{DA87}. It has been shown in recent work by Fraca and
Haddad~\cite{FH15} that the decision problems we consider in this
paper become substantially easier for such continuous Petri nets, with
reachability even being decidable in \P.

Finally, constraining the sequences of production rules applicable in
language generating devices is a classical topic in formal language
theory and commonly studied in the setting of controlled grammars. In
this context, valence grammars~\cite{Paun80} and blind counter
automata~\cite{Grei78} are closely related to our work and have led to
a large body of research, see e.g.~\cite{Hoog01,FS02,BZ13} and the
references therein. Valence grammars over the monoid $(\zint, +)$ are
context-free grammars in which every production rule is tagged with an
integer, and a word is generated by the grammar whenever the sum of
all integers that tag the production rules in its derivation equals
zero. The results of this paper allow one to obtain
complexity-theoretic upper bounds for deciding emptiness in (a
generalization of) valence grammars over this monoid.



\section{Preliminaries}
\label{sec:def}

In this section, we provide basic definitions that we rely on in
this paper. First, we introduce some general notation and standard
definitions related to Presburger arithmetic and formal language
theory. We then introduce the class of context-free commutative
grammars that we study in this paper and recall some known results
about this class from the literature.

\subsection{General Notation.} 
In the following, $\zint$ and $\nat$ are the sets of integers and
natural numbers (non-negative integers), respectively, and $\nat^d$ and $\zint^d$ are the set
of dimension-$d$ vectors in $\nat$ and $\zint$, respectively.  If not
stated otherwise, all numbers in this paper are assumed to be encoded
in binary. For $a,b\in \zint$ such that $a<b$, we denote by $[a,b]$
the set $\{a, a+1,\ldots, b\}$. As an abbreviation, $[d]$ denotes
$[1,d]$. For $\vec{v} \in \zint^d$ we write $\vec{v}(i)$ for the
$i$-th component of $\vec{v}$ for $i\in [d]$. Let $z\in \zint$, we
denote by $\vec{z}$ the vector in any dimension which has value $z$ in
each of its components. Given two vectors $\vec{v}_1, \vec{v}_2 \in
\zint^d$, we write $\vec{v}_1 \ge \vec{v}_2$ if and only if for all
$i\in [d]$, $\vec{v}_1(i) \ge \vec{v}_2(i)$.  Given a vector $\vec{v}
\in \zint^d$ and a set $R \subseteq [d]$, by $\vec{v}_{|R}$ we denote
the vector which coincides with $\vec{v}$ except for components from
$R$ which are \emph{reset} to zero, \emph{i.e.},
\[
\vec{v}_{|R}(i) \defeq 
\begin{cases}
   \vec{v}(i) & \text{if } i\not\in R\\
   0 & \text{otherwise}.
\end{cases}
\]
We call $|R$ the \emph{reset operator}.

\subsection{Presburger Arithmetic.}\label{ssec:presburger}
The first-order theory of the structure $\langle \nat, 0, 1, +,
\ge\rangle$, i.e., quantified linear arithmetic over natural numbers,
is commonly known as \emph{Presburger arithmetic (PA)}. The size
$|\Phi|$ of a PA-formula $\Phi$ is the number of symbols required to
write it down. Two fragments of Presburger arithmetic with a fixed
number of quantifier alternations are relevant to us in this paper.
\begin{proposition}\label{prop:pa-complexity}
  The existential $\Sigma_1$-fragment of Presburger arithmetic is
  \NP-complete~\cite{BT76}. Validity in the $\Pi_2$-fragment of PA,
  i.e.\ its restriction to a $\forall^*\exists^*$-quantifier
  prefix, is $\coNEXP$-complete~\cite{Grae89,Haa14}.
\end{proposition}

Given a PA-formula $\Phi(x_1,\ldots,x_d)$ in $d$ free variables, we
define
\begin{align*}
  \eval{\Phi(x_1,\ldots,x_d)} & \defeq \{ (n_1,\ldots,n_d)\in \nat^d : 
  \Phi(n_1/x_1,\ldots,n_d/x_d) 
  \text{ is valid}\}.
\end{align*}
Here, $\Phi(n_1/x_1,\ldots,n_d/x_d)$ is obtained from $\Phi$ by
replacing every $x_i$ with $n_i$; we also write
$\Phi(\vec{n}/\vec{x})$ as a shorthand for replacing the components of
$\vec{x}$ with the respective components of $\vec{n}$ in $\Phi$. For
notational convenience, we sometimes use vectors of vectors of
first-order variables and denote them by bold capital letters, e.g.,
$\vec{X}=(\vec{x}_1,\ldots,\vec{x}_k)$, where the $\vec{x}_i$ are
vectors of first-order variables.

A set $M\subseteq \nat^d$ is \emph{PA-definable} if there exists a PA
formula $\Phi(x_1,\ldots,x_d)$ such that
$M=\eval{\Phi(x_1,\ldots,x_d)}$. Recall that a result due to Ginsburg
\& Spanier states that PA-definable sets coincide with
\emph{semi-linear sets}~\cite{GS66}. A subset of $M\subseteq \nat^d$
is linear if there exist $\vec{b}\in \nat^d$ and $Q=\{
\vec{q}_1,\ldots, \vec{q}_n\}\subseteq \nat^d$ such that
\[
M = L(\vec{b}, Q) \defeq  \vec{b} + \{ \lambda_1 \cdot \vec{q}_1 + \cdots + 
\lambda_n \cdot \vec{q}_n : \lambda_i \in \nat\};
\]
semi-linear sets are finite unions of linear sets and are closed under
all Boolean operations~\cite{GS66}.

In this paper, we sometimes wish to define subsets of $\zint^d$ via
formulas of Presburger arithmetic. Clearly, any integer $z$ can be
represented as the difference of two natural numbers $x$ and
$y$. Hence, the homomorphism $h: \nat^{2d} \to \zint^d$ defined as
\[
h: (x_1,y_1,\ldots,x_n,y_n) \mapsto (x_1-y_1,\ldots,x_n-y_n)
\]
can be lifted in order to uniquely assign
a subset of $\zint^d$ to every subset of $\nat^{2d}$.
Thus, whenever it is convenient for us, we may with
no loss of generality interpret some open variables of formulas of
Presburger arithmetic in the integers (we will explicitly mention such cases).

\subsection{Formal Languages.}
Let $\Sigma=\{a_1,\ldots,a_m\}$ be a finite alphabet. The free monoid
generated by $\Sigma$ is denoted by $\Sigma^*$, and by $\Sigma^\odot$
we denote the free commutative monoid generated by $\Sigma$.
Elements of $\Sigma^*$ are words, i.e., finite sequences of elements
from $\Sigma$, with the usual concatenation operation $\cdot$.
Elements of $\Sigma^\odot$ are commutative words; we treat them
as mappings of the form $\Sigma \to \nat$,
or, equivalently, as vectors from $\nat^m$ with component-wise addition.
The empty word is denoted by~$\varepsilon$. Given $w\in \Sigma^* \cup
\Sigma^\odot$ and $a\in \Sigma$, $|w|_a$ denotes the number of times
$a$ occurs in the (usual or commutative) word~$w$.
We interchangeably use different equivalent ways in
order to represent a word $w\in \Sigma^\odot$. For $j\in [m]$ let
$i_j=|w|_{a_j}$; we equivalently write $w$ as
$w=a_1^{i_1}a_2^{i_2}\cdots a_m^{i_m}$, $w=(i_1,i_2,\ldots,i_m)\in
\nat^m$ or $w:\Sigma \to \nat$ with $w(a_j)=i_j$, whichever is most
convenient. Given $v,w\in \Sigma^\odot$, we write $v+w$
to denote the sum of $v$ and $w$.  Given
$w\in \Sigma^*$, we denote by $\pi(w)\in \Sigma^\odot$ its Parikh
image, i.e., $\pi(w)\defeq (|w|_{a_1},\ldots,|w|_{a_m})$.

Viewing commutative words as elements of $\nat^m$ allows us to employ
them inside formulas of Presburger arithmetic. In particular, given a
vector $\vec{x}=(x_1,\ldots, x_m)$ of first-order variables and a
commutative word $w = ({i_1}, \dots, {i_m}) \in \Sigma^\odot$, then
$\vec{x}=w$ abbreviates $\bigwedge_{1\le j\le m} x_j = i_j$.

\subsection{Context-Free Commutative Grammars with Integer Counters
and Resets.}\label{ssec:zcfcgr}

The main objects studied in this paper are derived from a general
class of context-free commutative grammars equipped with integer
counters\footnote{In the literature, such counters are often also
  called \emph{blind counters}.} which can be reset, incremented or
decremented when production rules are applied. Formally, these grammars
are defined as follows.
\begin{definition} 
  A \emph{context-free commutative grammar with integer counters and
    resets (\ZCFGr)} is a quadruple $\mathcal{G}=(N, C, P, S)$ where
  \begin{itemize}
  \item $N$ is a finite alphabet of \emph{non-terminal symbols};
  \item $C$ is a finite set of \emph{counters};
  \item $P \subseteq N \times 2^C \times \zint^C \times N^\odot$ is a
    finite set of \emph{production rules}; and
  \item $S \in N$ is the \emph{axiom}.
  \end{itemize}
\end{definition}

We often write $p \in P$ as a tuple of elements indexed by $p$, i.e.,
as $p = (a_p, R_p, \vec{z}_p, w_p)$. Informally, the
production $p$ can be applied whenever the non-terminal $a_p$ is
available; it then resets the counters specified by $R_p$ and adds
$\vec{z}_p$ to all counters while producing non-terminal symbols
$w_p$. Formally, let $C(\mathcal{G}) \defeq N^\odot \times \zint^C$ be
the \emph{set of configurations of $\mathcal{G}$}. Given
configurations $(s, \vec{u}), (t, \vec{v})\in C(\mathcal{G})$ and
$p\in P$, we write $(s, \vec{u})\xrightarrow{p}_\mathcal{G} (t,
\vec{v})$ if there is some $w\in N^\odot$ such that
\begin{itemize}
\item $s=w+a_p$,
\item $t = w + w_p$; and
\item $\vec{v} =\vec{u}_{|R} + \vec{z}_p$.
\end{itemize}
We write $(s, \vec{u}) \rightarrow_{\mathcal{G}} (t, \vec{v})$
whenever $(s,\vec{u})\xrightarrow{p}_\mathcal{G} (t, \vec{v})$ for
some $p\in P$.

A \emph{run} is a word $\gamma=p_1\cdots p_n\in P^*$ such that there
exists a finite sequence of configurations $\varrho: c_0 c_1 \cdots
c_n$ such that $c_i\xrightarrow{p_{i+1}}_{\mathcal{G}} c_{i+1}$ for
all $0\le i<n$, and we write $c_0 \xrightarrow{\gamma}_\mathcal{G}
c_n$ in this case. Furthermore, we write $c \rightarrow^*_\mathcal{G} c'$
if there is a run $\gamma \in P^*$ such that
$c\xrightarrow{\gamma}_\mathcal{G} c'$.
We drop the subscript $\mathcal{G}$ if it is clear from the context.
Given $\vec{u} \in \zint^C$,
the \emph{reachability set starting from $\vec{u}$} is defined as
\[
  \mathit{reach}(\mathcal{G},\vec{u}) = \{ \vec{v}\in \zint^C :
  (S, \vec{u}) \rightarrow^*_{\mathcal{G}} (t, \vec{v}) \text{ for some } 
  t\in N^\odot \}.
\]



\begin{rmk}
  Context-free (commutative) grammars are commonly used as language
  acceptors or generators. In our setting, when restricting counter
  updates to $\nat$ (i.e., when $P \subseteq N \times 2^C \times \nat^C \times N^\odot$),
  we may view \ZCFGr \ as generators of languages
  over $C^\odot$.
\end{rmk}

In this paper, we study the computational complexity of deciding
reachability, coverability and inclusion in \ZCFGr.
\vspace*{0.25cm}
\problemx{\ZCFGr \ Reachability/Coverability/Inclusion}
  {\ZCFGr \ $\mathcal{G}$, $\mathcal{H}$ over the same set of counters $C$
    and configurations 
    $(s,\vec{u}), (t,\vec{v})\in C(\mathcal{G})$,
    $\vec{v},\vec{v}' \in \zint^C$.}
  {\emph{Reachability:} Is there a run $(s, \vec{u})\rightarrow^*_\mathcal{G} (t, \vec{v})$?\\
    \emph{Coverability:} Is there a $\vec{z}\in \zint^C$ such that
    $(s, \vec{u}) \rightarrow^*_\mathcal{G} (t, \vec{z})$ and 
    $\vec{z}\ge \vec{v}$?\\
    \emph{Inclusion:} Does $\mathit{reach}(\mathcal{G},\vec{u}) \subseteq 
    \mathit{reach}(\mathcal{H}, \vec{v})$ hold? 
  }
\vspace*{0.25cm}
We also study and discuss natural subclasses of \ZCFGr \ where we
restrict the use of reset operations or the set of productions of
the grammar. A \ZCFGr \ $\mathcal{G} = (N, C, P, S)$ is an
\begin{itemize}
\item \emph{integer vector addition system with states (\ZVASS)} if $P
  \subseteq N \times \{\emptyset\} \times \zint^C \times (N \cup \{ \varepsilon \})$;
\item \emph{integer vector addition system (\ZVAS)} if $\mathcal{G}$
  is a \ZVASS \ and $N = \{ S \}$.
\end{itemize}
\ZVASS\ are obtained from \ZCFGr \ by restricting the grammar to be
left-linear and by disallowing resets. We use them in order to obtain
stronger lower bounds. Left-linear context-free grammars are known to
recognize regular languages, and equivalently, \ZVASS \ can be seen as
finite-state automata equipped with integer counters. Formalized in
this manner, it is easier to see that classical vector addition
systems with states (\VASS) can be recovered from the definition of
\ZVASS \ by restricting the set of configurations to $(N \cup \{
\varepsilon \}) \times \nat^C$ and adjusting the definition of
$\rightarrow_{\mathcal{G}}$ appropriately. This justifies
the term~``\ZVASS''. Note that in \ZVASS, we restrict commutative words in
configurations to length at most one, and we restrict the
reachability problem accordingly.

Finally, note that a \ZVAS \ ${\mathcal A} = (\{ S \}, C, P, S)$ can
simply be represented by a matrix $A \in \zint^{d\times k}$ where $d =
|C|$ and $k = |P|$ and $A$ is the matrix whose columns are $\vec{z}_p$
for $p \in P$. The matrix $A$ has the property that for all $\gamma
\in P^*$, $(S, \vec{u}) \xrightarrow{\gamma} (S, \vec{u} + A \cdot
\pi(\gamma))$. Consequently, reachability in \ZVAS\ and all classes
subsuming \ZVAS\ is \NP-hard, which can be shown by a reduction from
the feasibility problem of a system of linear Diophantine equations
$A\cdot \vec{x} = \vec{u}, \vec{x} \ge
\vec{0}$.\label{anc:matrix-reachability} This problem is known to be
\NP-hard even when numbers are encoded in unary~\cite{GJ79}. We have
that $A\cdot \vec{x} = \vec{u}, \vec{x} \ge \vec{0}$ is valid if and
only if $(S, -\vec{u}) \rightarrow^*_{\mathcal{A}} (S, \vec{0})$ in
the corresponding \ZVAS.

\subsubsection*{Relationship to Communication-Free Petri Nets.}\label{ssec:relationship}
As already stated in the introduction, context-free commutative
grammars are closely related to communication-free Petri nets. For
inter-reducibility results, see, e.g.,~\cite{MW15}. For the sake of
completeness, here we briefly state the relationship on an
informal level. 

Viewed in our framework, context-free commutative grammars are
\ZCFGr\ whose integer counters can only be incremented and thus
correspond to \emph{terminal symbols}. Context-free commutative
grammars correspond to communication-free Petri nets by viewing the
set of non-terminal and terminal symbols as the set of places of the
Petri net.
Similarly, \ZCFGr\ correspond to communication-free
Petri nets which are additionally equipped with special places that
can take a possibly negative number of tokens and with special arcs
that can set the number of tokens on those counters to zero.
All upper bounds for context-free commutative
grammars carry over to communication-free Petri nets, and all lower
bounds for communication-free Petri nets carry over to context-free
commutative grammars, see for example~\cite{MW15}.


\section{Reachability and Coverability in \ZCFGr}\label{sec:zvass}

In this section, we consider the reachability and coverability problem
for \ZCFGr. We begin by showing that reachability and coverability are
logarithmic-space interreducible; such a reduction is not known for
general Petri nets and cannot exist for Petri nets equipped with reset
operations since reachability in such nets is undecidable whereas
coverability is decidable~\cite{DFS98}. Subsequently, we show that
reachability and hence coverability in \ZCFGr\ is \NP-complete by
showing that the reachability relation is definable by a sentence in
existential Presburger arithmetic of polynomial size.

\subsection{Reachability and Coverability are Interreducible}\label{ssec:interrdecuction}

Here, we show that reachability and coverability are logarithmic-space
interreducible in \ZCFGr\ and all of the subclasses we introduced in
Section~\ref{ssec:zcfcgr}. Thanks to this observation, all lower and upper
bounds for reachability carry over to coverability, and \emph{vice
  versa}.
\begin{theorem}\label{equiv}
  Reachability and coverability are logarithmic-space interreducible
  in each of the classes \ZCFGr, \ZVASS \ and \ZVAS. The reduction
  doubles the number of counters.
\end{theorem}

\begin{proof}
We first show how to reduce coverability to reachability. We adapt the
folklore construction used for reducing coverability in \VASS\ to
reachability in \VASS. This reduction adds extra transitions in order
to make the \VASS\ \emph{lossy}, i.e.\ transitions that allow counters
to be non-deterministically decremented at any time. To this end, let
${\mathcal G} = (N, C, P, S)$ be a \ZCFGr\ and $(s, \vec{u}), (t,
\vec{v})$ two configurations of $\mathcal G$. Define ${\mathcal H} =
(N, C, P', S)$ where $P' = P \cup \{ (S, \emptyset, -\vec{c}, S) :
c\in C \}$. Here $\vec{c}$ is the vector $\vec{c}: C \to \nat$
such that $\vec{c}(c)=1$ and $\vec{c}(c')=0$ for all $c'\neq c$. It is
easily seen that $(t, \vec{v})$ can be covered in $\mathcal G$
starting at $(s, \vec{u})$ if and only if there is a run $(s+S,
\vec{u}) \rightarrow^*_{\mathcal H} (t+S, \vec{v})$ in $\mathcal H$
due to the monotonicity of coverability.



We now show how to reduce reachability to coverability. Let ${\mathcal
  G} = (N, C, P, S)$ be a \ZCFGr\ and let $(s,\vec{u}),(t,\vec{v})\in
C(\mathcal{G})$. We construct a \ZCFGr\ $\mathcal{H} = (N, C \uplus
\tilde{C}, P', S)$ where $\tilde{C}\defeq \{ \tilde{c} : c \in C\}$
consists of an additional disjoint copy of $C$ and $P'$ contains a
production of the form $(a, R \cup \tilde{R}, (\vec{z}, -\vec{z}), w)$
whenever $(a, R, \vec{z}, w)$ is a production of $P$. The
\ZCFGr\ $\mathcal{H}$ therefore has the two following properties:
\begin{itemize}
\item starting from a configuration $(s, (\vec{u}, -\vec{u}))$, any
  configuration reached is of the form $(t, (\vec{v}, -\vec{v}))$; and
\item $(s, \vec{u})\rightarrow^*_{\mathcal{G}} (t,\vec{v})$ iff $(s,
  (\vec{u},-\vec{u}))\rightarrow^*_{\mathcal{H}} (t,
  (\vec{v},-\vec{v}))$.
\end{itemize}
These properties are easily shown by induction on the length of the
run. Consequently, the configuration $(t, \vec{v}, -\vec{v}))$ can be
covered starting at $(s, (\vec{u}, -\vec{u}))$ in $\mathcal H$ if and
only if there exists some $\vec{z}$ such that $(s, (\vec{u},
-\vec{u})) \rightarrow^*_{\mathcal H} (t, (\vec{z}, -\vec{z}))$ and
$\vec{z} \ge \vec{v}$ and $-\vec{z} \ge -\vec{v}$; i.e.\ if and only
if the configuration $(t, (\vec{v}, -\vec{v}))$ is actually reached in
$\mathcal H$. Consequently this is equivalent to $(t, \vec{v})$ being
reachable from $(s, \vec{u})$ in $\mathcal G$.

Finally, observe that all of the reductions described above
preserve the restrictions imposed on the subclasses of \ZCFGr\ and are
thus also valid for \ZVASS\ and \ZVAS.\qed
\end{proof}
%


\subsection{Reachability and Coverability in \ZCFGr\ are \NP-Complete}

As discussed in Section~\ref{ssec:zcfcgr}, reachability is already
\NP-hard for \ZVAS. In this section, we establish a matching upper
bound for \ZCFGr. One main idea for showing the upper bound is that
since there are no constraints on the values of the integer counters
along a run, a reset on a particular counter allows to forget any
information about the value of this counter up to this point, i.e., a
reset cuts the run. Hence, in order to determine the value of a
particular counter at the end of a run, we only need to sum up the
effects of the operations on this counter since the last occurrence of
a reset on this counter. Moreover, since addition and subtraction are
commutative, the order in which these effects occur is irrelevant.
That is, to determine whether a certain configuration on integer
counters is reached by a run, it suffices to consider the Parikh image
of this run.


Subsequently, we introduce a generalization of the notion of the
Parikh image of a run that, in effect, enables us to access the last
occurrence of a reset on a counter. This can be achieved by recording
the last occurrence of each production in $P = \{ p_1, \dots, p_k \}$,
some of which may not reset any counter at all. This idea leads to the
following unique decomposition of any run $\gamma \in P^+$ into
\emph{partial runs} $\gamma_1,\ldots, \gamma_{i_\ell}$ as
\[
  \gamma = \gamma_1 p_{i_1} \gamma_2 p_{i_2} \cdots \gamma_l p_{i_\ell}
\]
for some $\ell\le k$ such that all $i_j$ are pairwise distinct and for
all $j \in [\ell]$, $\gamma_j \in \{ p_{i_j}, \dots, p_{i_\ell}\}^*
$. This decomposition simply keeps track of the last occurrence of
each production used in $\gamma$. For instance for $P = \{ a,b,c,d,e
\}$, the word $\gamma = aaebaeabba$ can uniquely be decomposed as
$(aaeba)e(ab)b()a$.
This decomposition is formalized in the following definition as the
\emph{generalized Parikh image} of a word.
By $\mathfrak{S}_k$ we denote the permutation group on $k$
symbols, and we sometimes treat its elements as vectors of
$\nat^k$.

\begin{definition}
\label{def}
Let $\mathcal{G} = (N, C, P, S)$ be a \ZCFGr \ with $P = \{p_1, \dots,
p_k \}$. A triple $(\vec{A},\sigma, m) =
(\vec{\alpha}_1,\vec{\alpha}_2,\dots, \vec{\alpha}_k, \sigma, m) \in
(\nat^k)^{k} \times \mathfrak{S}_k \times [k]$ is \emph{a generalized
  Parikh image of $\gamma \in P^+$} if there exists a decomposition
\[
\gamma = \gamma_m p_{\sigma(m)} \gamma_{m+1}p_{\sigma(m+1)} \cdots
\gamma_k p_{\sigma(k)}
\] 
such that
\begin{enumerate}[(i)]
\item for all $m\le i \le k$, $\gamma_i \in \{p_{\sigma(i)}, \dots, p_{\sigma(k)}\}^*$; and
\item for all $1\le i < m$, $\vec{\alpha}_i = \vec{0}$, and for all $m
  \le i \le k$, $\vec{\alpha}_i = \pi(\gamma_i)$.
\end{enumerate}
We denote by $\Pi(\gamma)$ the set of all generalized Parikh images of
a word $\gamma \in P^+$.
\end{definition}
This definition formalizes the intuition, combining the decomposition
described above with some padding by dummy vectors for
productions that do not occur in $\gamma$, in order to obtain canonical
objects of \emph{uniform size}. Even though generalized Parikh images
are not unique, two generalized Parikh images of the same word differ
only in the order of productions that do not appear
in $\gamma$.  For instance, if $P = \{ a,b,c,d,e \}$, the word
$\gamma = aaebaeabba$ has two generalized Parikh images: they agree
on $\vec{\alpha}_1 = \vec{\alpha}_2 = (0,0,0,0,0)$, $\vec{\alpha}_3 =
(3,1,0,0,1)$, $\vec{\alpha}_4 = (1,1,0,0,0)$, $\vec{\alpha}_5 =
(0,0,0,0,0)$ and $\sigma(3) = 5$, $\sigma(4) = 2$, $\sigma(5) = 1$,
and $m = 3$, and only differ on $\sigma(1)$ and $\sigma(2)$ that can
be $3$ and $4$, or $4$ and $3$, respectively.

Generalized Parikh images can now be applied to reachability in
\ZCFGr\ as follows: the counter values at the end of a run $\gamma\in
P^+$, starting from an initial configuration $(s, \vec{u})$, are fully
determined by a generalized Parikh image of $\gamma$, as shown in the
next lemma.
Recall that, if $P = \{ p_1, \ldots, p_k \}$, then
each production $p_\ell$ resets the counters in the set $R_{p_\ell} \subseteq C$.
Subsequently, for $i\in [1,k]$ we write
\[
R_i = R_{p_{\sigma(i)}} \cup \dots \cup R_{p_{\sigma(k)}};
\] 
note that $R_i$ depends on the set of productions $P$ and on the
permutation~$\sigma$. Also, $R_{k+1}$ will
denote the empty set.

%
\begin{lemma}
  \label{lem:parikh-to-effect}
  Let $\mathcal{G} = (N, C, P, S)$ be a \ZCFGr \ with $P = \{ p_1,
  \dots, p_k \}$, $(s, \vec{u})$ and $(t, \vec{v})$ two configurations
  of $C(\mathcal{G})$ and $\gamma \in P^+$ such that $(s, \vec{u})
  \xrightarrow{\gamma} (t, \vec{v})$. Moreover, let
  $(\vec{\alpha}_1,\dots, \vec{\alpha}_k, \sigma, m)\in\Pi(\gamma)$ be
  a generalized Parikh image of $\gamma$. Then the following holds:
  \begin{align*}
    \vec{v} & = \vec{u}_{|R_m} + \sum_{i=m}^k \left[ \left( \sum_{p\in P} 
      \vec{\alpha}_i(p) \cdot \vec{z}_p \right)_{|R_i} + \left(\vec{z}_{p_{\sigma(i)}}\right)_{|R_{i+1}} \right].
  \end{align*}
\end{lemma}

\begin{proof}
The proof of the lemma formalizes the intuition given in the
introduction of this section: in order to determine the final counter
values at the end of the run, it is sufficient to only consider the
effects after the last reset has occurred on a particular counter.

Formally, let $\gamma = \gamma_m p_{\sigma(m)} \gamma_{m+1} \dots
\gamma_{k} p_{\sigma(k)}$ be the decomposition associated to the
generalized Parikh image $(\vec{\alpha}_1,\dots, \vec{\alpha}_k,
\sigma, m)$ of $\gamma$. Moreover, let $(s_m, \vec{u}_m), \dots, (s_k,
\vec{v}_k)$ and $(t_m, \vec{v}_m), \dots, (t_k, \vec{v}_k)$ be the
configurations such that for any $i \in [m,k]$,
\[(s_i, \vec{u}_i)
\xrightarrow{\gamma_i}_{\mathcal G} (t_i, \vec{v}_i)
\xrightarrow{p_{\sigma(i)}}_{\mathcal G} (s_{i+1}, \vec{u}_{i+1}),
\]
where $(s_m, \vec{u}_m) = (s, \vec{u})$ and $(s_{k+1}, \vec{u}_{k+1})
= (t, \vec{v})$.

We prove the following statement by induction on $j \in [m,
  k]$:
\begin{align}
{\vec{u}_j}_{|R_j} &= \vec{u}_{|R_m} + \sum_{i = m}^{j-1} \left[ \left( \sum_{p\in P} \vec{\alpha}_i(p) \cdot \vec{z}_p \right)_{|R_i} + \left(\vec{z}_{p_{\sigma(i)}}\right)_{|R_{i+1}} \right] \label{firsteq}\\
{\vec{v}_j}_{|R_j} &= \vec{u}_{|R_m} + \sum_{i = m}^{j-1} \left[ \left( \sum_{p\in P} \vec{\alpha}_i(p) \cdot \vec{z}_p \right)_{|R_i} + \left(\vec{z}_{p_{\sigma(i)}}\right)_{|R_{i+1}} \right] + \left( \sum_{p\in P} \vec{\alpha}_j(p) \cdot \vec{z}_p \right)_{|R_j}.
\end{align}

\emph{Base case $j=m$:} Equation~(\ref{firsteq}) is obvious. Since only resets
on components $c \in R_m$ occur in $\gamma_m$ by definition of the
decomposition, and since addition is commutative and associative, only
the number of times each production appears is important. Hence
\begin{align*}
{\vec{v}_m}_{|R_m} - {\vec{u}_m}_{|R_m} = \left( \sum_{p\in P} |\gamma_m|_{p} \cdot {\vec{z}_p} \right)_{|R_m} = \left( \sum_{p\in P} \vec{\alpha}_m(p) \cdot {\vec{z}_p} \right)_{|R_m}.
\end{align*}

\emph{Induction step $j > m$:} The configuration $(s_j, \vec{u}_j)$ is
obtained from the configuration $(t_{j-1}, \vec{v}_{j-1})$ using the
production $p_{\sigma(j-1)}$, therefore
\begin{align*}
\vec{u}_j = ( {\vec{v}_{j-1}})_{|R_{p_{\sigma(j-1)}}} + \vec{z}_{p_{\sigma(j-1)}},
\end{align*}
which leads to
\begin{align*}
{\vec{u}_j}_{|R_j} &= (( {\vec{v}_{j-1}})_{|R_{p_{\sigma(j-1)}}} + \vec{z}_{p_{\sigma(j-1)}} )_{|R_j} \\
&= ( {\vec{v}_{j-1}})_{|R_{p_{\sigma(j-1)}} \cup R_j} + (\vec{z}_{p_{\sigma(j-1)}})_{|R_j}  \\
&= (\vec{v}_{j-1})_{|R_{j-1}} + (\vec{z}_{p_{\sigma(j-1)}})_{|R_j}  \\
&= \vec{u}_{|R_m} + \sum_{i = m}^{j-2} \left[ \left( \sum_{p\in P} \vec{\alpha}_i(p) \cdot \vec{z}_p \right)_{|R_i} + \left(\vec{z}_{p_{\sigma(i)}}\right)_{|R_{i+1}} \right] + \\
&~~~~~~~~~~~~\left( \sum_{p\in P} \vec{\alpha}_{j-1}(p) \cdot \vec{z}_p \right)_{|R_{j-1}} + {\vec{z}_{p_{\sigma(j-1)}}}_{|R_j} \\
&= \vec{u}_{|R_m} + \sum_{i = m}^{j-1} \left[ \left( \sum_{p\in P} \vec{\alpha}_i(p) \cdot \vec{z}_p \right)_{|R_i} + \left(\vec{z}_{p_{\sigma(i)}}\right)_{|R_{i+1}} \right].
\end{align*}
In a similar way, the configuration $(t_j, \vec{v}_j)$ is obtained
from the configuration $(s_{j}, \vec{u}_j)$ by applying the partial
run $\gamma_j$, which only resets counters in $R_j$. Therefore,
\begin{align*}
{\vec{v}_j}_{|R_j} = \left( \vec{u}_j + \sum_{p\in P} \vec{\alpha}_j(p) \cdot \vec{z}_p \right)_{|R_j}.
\end{align*}
%
%
The statement of the lemma now follows from taking $j=k+1$ in
Equation~(\ref{firsteq}).\qed
\end{proof}

Thus, in order to decide reachability in \ZCFGr, it suffices to find a
suitable way to reason about generalized Parikh
images. In~\cite{VSS05}, Verma et al.\ show how to construct in
polynomial time an existential Presburger formula representing the
Parikh image of the language of a context-free grammar. We
generalize this construction to generalized Parikh images of
\ZCFGr. First, let us state the result from~\cite{VSS05} using the
terminology of this paper.

\begin{proposition}{\cite[Thm.\ 4]{VSS05}}\label{lem1}
Given a \ZCFGr \ $\mathcal{G} = (N, C, P, S)$ with $|N| = n$ and $|P|
= k$, one can compute in polynomial time an existential Presburger
formula $\varphi_{\mathcal{G}}(\vec{s}, \vec{t}, \vec{\alpha})$ where
$\vec{s}$, $\vec{t}$ are $n$-tuples and $\vec{\alpha}$ is a $k$-tuple
of first-order variables such that for all $s, t$ in $N^\odot$ and
$\vec{\alpha} \in \nat^k$, the following are equivalent:
\begin{itemize}
\item $(s, t, \vec{\alpha}) \in \eval{\varphi_{\mathcal{G}}}$
\item there is a run $\gamma \in P^*$ with $\pi(\gamma) =
  \vec{\alpha}$ such that for any $\vec{u} \in \zint^C$, $(s, \vec{u})
  \xrightarrow{\gamma}_{\mathcal G} (t, \vec{v})$ for some $\vec{v}
  \in \zint^C$.
\end{itemize}
\end{proposition}
In other words, a model $(s, t, \vec{\alpha})$ of the formula
$\varphi_{\mathcal G}$ asserts that $\vec{\alpha}$ is the Parikh image
of a valid run between the commutative words $s$ and $t$. To implement
the definition of generalized Parikh image in Presburger arithmetic,
it is now sufficient to guess the intermediate words and ``connect''
them using formulas $\varphi_{\mathcal G}$. Subsequently, whenever we
define a permutation $\sigma$ in Presburger arithmetic, we write
$\vec{\sigma}$ for the corresponding vector of first-order variables
defining the respective components of the vector representation of
$\sigma$.
\begin{lemma}
  \label{gpi}
  Let $\mathcal{G} = (N, C, P, S)$ be a \ZCFGr. There exists a
  polynomial-time computable existential Presburger formula
  $\Psi_\mathcal{G}(\vec{s}, \vec{t}, \vec{A}, \vec{\sigma}, m)$
  defining the generalized Parikh images of runs of $\mathcal{G}$ from
  $s$ to $t$.
\end{lemma}
\begin{proof}
  Let $P = \{ p_1, \dots, p_k \}$. Subsequently, we identify
  productions $p_i \in P$ with their index $i$. This enables us to
  write atomic formulas such as $x = p_i$, where $x$ is a first-order
  variable. Given $p = (a_p, R_p, \vec{z}_p, w_p) \in P$, we denote by
  $a_p$ and $w_p$ the corresponding vectors from $\nat^N$ as constant
  terms in the logic, and, similarly, by $\vec{z}_p$ the corresponding constant vector from
  $\zint^k$. Remember that equalities between vectors or commutative
  words in Presburger arithmetic abbreviates the conjunction of formulas
  expressing equality of their components.

  The formula we construct has vectors of free variables $\vec{s}$ and
  $\vec{t}$ for the starting and ending non-terminal commutative
  words; $\vec{\alpha_1}, \dots, \vec{\alpha_k}$ gathered in the
  matrix of first order variables $\vec{A}$, $\vec{\sigma}=(\sigma_1,
  \dots, \sigma_k)$, and a variable $m$ that encode a generalized
  Parikh image. First, we construct a formula $\varphi_{\text{perm}}$
  asserting that $\vec{\sigma}$ is a permutation on the set $[k]$:
\begin{align*}
  \varphi_{\text{perm}}(\vec{\sigma}) \defeq
  \bigwedge_{i\in [k]}\left( 1 \le \sigma_i \le k \ \wedge
  \bigwedge\nolimits_{j \in [k]} i \neq j \rightarrow \sigma_i \neq \sigma_j \right).
\end{align*}
This formula has size $O(k^2)$ and is thus polynomial in
$|\mathcal{G}|$. Now we must ``compute'' the $k$ partial runs, but first
we have to ``guess'' the starting and ending words of $N^\odot$ of
each of these partial runs, in order to use the formula from
Lemma~\ref{lem1}. Let $\vec{S} =(\vec{s}_1,\ldots,\vec{s}_k)$ and
$\vec{T}=(\vec{t}_1,\ldots,\vec{t}_k)$ and define
\begin{multline*}
  \varphi_{\text{words}}(\vec{s}, \vec{t}, \vec{\sigma}, m, \vec{S}, \vec{T}) \defeq
    \vec{s}_1 = \vec{s} \ \wedge \\ \wedge
  \bigwedge_{p\in P} (\sigma_k = p \rightarrow \vec{t}_k(a_p) > 0 \wedge \vec{t} = \vec{t}_k - a_p + w_p) \ 
    \wedge \\ \wedge
  \bigwedge_{1 \le i < k}  \bigg[(i < m \rightarrow \vec{s}_i = \vec{t}_i \wedge 
    \vec{t}_i= \vec{s}_{i+1}) \ \wedge  \\ \wedge
  \bigg(m \le i \rightarrow \bigg(\bigwedge_{p \in P} \sigma_i = p \rightarrow \vec{t}_i(a_p) > 0 
    \wedge \vec{s}_{i+1} = \vec{t}_i - a_p + w_p\bigg)\bigg)\bigg].
\end{multline*}
Here, $m$ is used as in Definition~\ref{def},
and $\vec{t}_k(a_p)$ and $\vec{t}_i(a_p)$ denote components of $\vec{t}_k$
and $\vec{t}_i$, respectively, whose index coincides with the index of
the only non-zero entry in the constant vector~$a_p \in \nat^N$.
The two first lines
enforce that the run is going from $\vec{s}$ to $\vec{t}$. The third
line imposes $\vec{s}_m = \vec{s}$, and the last one ensures that the
production $p_{\sigma_i}$ can be applied from $\vec{t}_i$ and reaches
$\vec{s}_{i+1}$. We can now express that the $k$ partial runs have Parikh
images $\vec{\alpha}_i$ and are connecting $\vec{s}_i$ with
$\vec{t}_i$, and that the production $p_{\sigma_i}$ is not occurring
afterwards in the decomposition.
\begin{multline*}
   \varphi_\text{runs}(\vec{\sigma}, m, \vec{A}, \vec{S}, \vec{T}) \defeq
   \bigwedge_{ i\in [k]} \bigg[ (i < m \rightarrow \vec{\alpha}_i = \vec{0}) \wedge
\\ \wedge \left(m \le i \rightarrow \left( \varphi_{\mathcal{G}}(\vec{s}_i,
   \vec{t}_i, \vec{\alpha}_i) \wedge
   \bigwedge_{1 \le j < i} \bigwedge_{p \in P} p = \sigma_j \rightarrow \vec{\alpha}_i(p) = 0 \right) \right)\Bigg].
\end{multline*}
In summary, $\varphi_{\text{perm}}$, $\varphi_{\text{words}}$ and
$\varphi_{\text{runs}}$ enforce the constraints from
Definition~\ref{def}. Putting everything together yields:
\begin{multline*}
  \Psi_\mathcal{G}(\vec{s}, \vec{t}, \vec{A}, \vec{\sigma}, m) \defeq \exists \vec{S},\vec{T}.\, \\
  1 \le m \le k \wedge
  \varphi_\text{perm}(\vec{\sigma}) \wedge
  \varphi_\text{words}(\vec{s}, \vec{t}, \vec{\sigma}, m, \vec{S}, \vec{T}) \wedge
  \varphi_\text{runs}(\vec{\sigma}, m, \vec{A}, \vec{S}, \vec{T}). 
\end{multline*}
Note that the size of $\Psi_\mathcal{G}(\vec{s}, \vec{t}, \vec{A},
\vec{\sigma}, m)$ is polynomial in $|\mathcal{G}|$. \qed
\end{proof}
By combining $\Psi_\mathcal{G}$ with Lemma~\ref{lem:parikh-to-effect},
we obtain the main theorem of this section. Subsequently, $\vec{u}$
and $\vec{v}$ are interpreted as vectors over the integers
(and not over the naturals); the details are as discussed
previously in Section~\ref{ssec:presburger}.

%

\begin{theorem}
  \label{cor:reachability-presburger}
  Let ${\mathcal G}$ be a \ZCFGr. There exists a polynomial-time
  computable existential Presburger formula 
  $\Phi_{\mathcal{G}}(\vec{s},\vec{t},\vec{u},\vec{v},\vec{A},\vec{\sigma},
  m)$ such that for all $s, t$ in $N^\odot$, $\vec{u}, \vec{v} \in \zint^C$ and
  $(\vec{A}, \sigma, m) \in (\nat^k)^k \times \nat^k \times \nat$
  the following are equivalent:
  \begin{itemize}
  \item $(s,t,\vec{u},\vec{v},\vec{A}, \sigma, m)\in
    \eval{\Phi_{\mathcal{G}}}$,
  \item there is $\gamma\in P^+$ such that $(s, \vec{u})
    \xrightarrow{\gamma}_\mathcal{G}(t, \vec{v})$ and $(\vec{A},
    \sigma, m)$ is a generalized Parikh image of $\gamma$.
  \end{itemize}
  In particular, reachability and coverability in \ZCFGr\ are
  \NP-complete.
\end{theorem}
\begin{proof}
  Thanks to the characterization of generalized Parikh images via
  $\Psi_{\mathcal{G}}$ obtained from Lemma~\ref{gpi}, it suffices to
  show that the equation obtained in Lemma~\ref{lem:parikh-to-effect}
  can be encoded in Presburger arithmetic. For any $c\in C$, this
  equation can be rewritten as follows:
  \begin{align*}
    \vec{v}(c) =& \vec{u}_{|R_m}(c) + \sum_{i=m}^k \left[ \left( \sum_{p\in P} \vec{\alpha}_i(p) \cdot \vec{z}_p \right)_{|R_i}(c) + 
      (\vec{z}_{p_{\sigma(i)}})_{|R_{i+1}}(c) \right] \\
    =& \lambda_{m,c}\cdot\vec{u}(c) + \sum_{i=1}^k \left[ \lambda_{i,c} \cdot \left( \sum_{p\in P}\vec{\alpha}_i(p) \cdot \vec{z}_p(c) \right) + 
      \lambda_{i+1,c} \cdot \vec{z}_{p_{\sigma(i)}}(c) \right]
  \end{align*}
  where
  \begin{align*}
    \lambda_{i,c} & =
    \left\{ \begin{array}{ll} 0 &\text{ if } c\in R_i \text{ or } i < m \\ 
      1 &\text{ otherwise}.
    \end{array} \right.
  \end{align*}
  Although it is easy to define $\lambda_{i,c}$ in Presburger
  arithmetic, the above equality is not a syntactically correct
  Presburger formula since the terms $\lambda_{m,c}\cdot\vec{u}(c)$
  and $\lambda_{i,c} \cdot \vec{\alpha}_i(p)$
  are not linear. To work around this problem, we therefore introduce
  intermediate variables $\beta_j^c$ and $\delta_j^c$ that enable us
  to handle the effect of resets in a step-wise fashion. Informally,
  we want these variables to satisfy the following conditions:
  \begin{align*}
    \beta_j^c = \lambda_{m,c}\cdot\vec{u}(c) + \sum_{i=1}^j \left[ \lambda_{i,c}\cdot  \left( \sum_{p\in P}\vec{\alpha}_i(p) \cdot \vec{z}_p(c) \right) + \lambda_{i+1,c} \cdot \vec{z}_{p_{\sigma(i)}}(c) \right]
\end{align*} 
for $j \in [0,k]$ and $c\in C$, and
\begin{multline*}
  \delta_j^c = \lambda_{m,c}\cdot\vec{u}(c) + \sum_{i=1}^{j-1}
   \left[ \lambda_{i,c}\cdot \left( \sum_{p\in P}\vec{\alpha}_i(p)
    \cdot \vec{z}_p(c) \right) + \lambda_{i+1,c} \cdot \vec{z}_{p_{\sigma(i)}}(c) \right]+\\ +
  \lambda_{j,c}\cdot\sum_{p\in P}\vec{\alpha}_j(p) \cdot \vec{z}_p(c)
\end{multline*}
for $j \in [1,k]$ and $c\in C$.
This approach is formalized in the formula $\varphi_\text{counters}$
below. First, remember that $R_i \defeq R_{p_{\sigma(i)}} \cup \dots
\cup R_{p_{\sigma(k)}}$, and, therefore, for any $c \in C$:
\begin{align*}
  c \in R_i &\iff c \in R_{p_{\sigma(i)}} \cup \dots \cup R_{p_{\sigma(k)}} \\
  &\iff \bigvee_{i \le j \le k} c \in R_{p_{\sigma(j)}} \\
  &\iff \bigvee_{i \le j \le k} \bigvee_{d\in R_{p_{\sigma(j)}}} d = c.
\end{align*}
We therefore introduce the notation $c \in R_x$, where $x$ can be a first-order
variable, as an abbreviation for the following formula:
\begin{align*}
\bigvee_{j=1}^k ( j \ge x ) \wedge
        (\bigwedge_{\ell = 1}^k (\ell = \sigma_j
        \rightarrow \bigvee_{d \in R_{p_\ell}} d = c)).
\end{align*}
Note that formulas of the form $\bigwedge_{\ell=1}^k \ell = \sigma_j
\rightarrow \dots$ are used when we need to use $\sigma_j$ as an index
(which would not be correct since $\sigma_j$ is a first-order
variable).
To improve readability, we also write $\lambda_{i,c} = 0$ to denote the formula
$c \in R_i \vee i < m$, and $\lambda_{i,c} = 1$
to denote $c \notin R_i \wedge i \ge m$.
Now, $\varphi_\text{counters}$ can be defined as follows:
\begin{multline*}
  \varphi_\text{counters} (\vec{s}, \vec{t}, \vec{u}, \vec{v}, \vec{A}, \vec{\sigma}, m) \defeq 
    \exists \vec{B}. \exists \vec{D}. \\
  \bigwedge_{c \in C} \bigg\{ (\lambda_{m,c} = 0 \rightarrow \beta_0^c = 0) 
    \wedge (\lambda_{m,c} = 1 \rightarrow \beta_0^c = \vec{u}(c)) \ \wedge \\
    \wedge \ \vec{v}(c) = \beta_k^c \ \wedge \\
    \wedge \ \bigwedge_{j=1}^k \bigg[ 
    (\lambda_{j,c} = 0 \rightarrow \delta_j^c = \beta_{j-1}^c) \wedge \\
    \wedge (\lambda_{j,c} = 1 \rightarrow \delta_j^c = \beta_{j-1}^c 
      +  \sum_{p\in P}\vec{\alpha}_j(p) \cdot \vec{z}_p(c)) \ \wedge \\
    \wedge (\lambda_{j+1,c} = 0 \rightarrow \beta_j^c = \delta_j^c) \\
    \wedge \bigg(\lambda_{j+1,c} = 1 \rightarrow \bigwedge_{\ell=1}^k (\ell = \sigma_j \rightarrow \beta_j^c = 
      \delta_j^c + z_{p_\ell}(c))\bigg) \bigg]\bigg\}.
  \end{multline*}
In this formula, the first line deals with $\beta_0^c$: it is
either $0$ or $\vec{u}(c)$ depending on whether $c \in R_m$, i.e.,
whether $c$ is reset at some point in the run. The second line gives
the desired value to $\vec{v}$. The four last lines compute
$\beta_j^c$ (respectively $\delta_j^c$) from $\delta_j^c$
(respectively $\beta_{j-1}^c$): if $\lambda_{j,c} = 0$ then nothing is added
to the sum.

Since satisfiability in existential Presburger arithmetic is
\NP-complete, this allows us to conclude that reachability in
\ZCFGr\ is in \NP\ and hence \NP-complete. By Theorem~\ref{equiv}, the
same result carries over to coverability in \ZCFGr. \qed
\end{proof}
Finally, we obtain as a corollary an existential Presburger formula
that defines the reachability set of a \ZCFGr\ that we will use in the
next section.
\begin{corollary}
 \label{cor:rs}
 Let $\mathcal{G}$ be a \ZCFGr. There exists a polynomial-time
 computable existential Presburger formula $\Phi^{\mathcal
   G}_\mathrm{rs}(\vec{u}, \vec{v})$ such that 
\[
(\vec{u}, \vec{v}) \in \eval{\Phi^{\mathcal G}_\mathrm{rs}} \iff \vec{v} \in \mathit{reach}(\mathcal{G},\vec{u}).
\]
\end{corollary}

\section{Inclusion for \ZCFGr}
\label{sec:inclusion}

In this section, we study inclusion problems for \ZCFGr\ and
subclasses thereof. We first remark that the general problem is
\coNEXP-complete. Subsequently, we show that the inclusion problem is
\ComplexityFont{\Pi_2^\P}-complete, even for the smallest subclass
\ZVAS. The proof of the lower bound requires us to develop a new
\ComplexityFont{\Pi_2^\P}-complete variant of the classic
\textsc{Subset Sum} problem, which we believe is a contribution of
independent interest.

\subsection{The General Case}
In this section, we show the following theorem.

\begin{theorem}\label{thm:inclusion-general}
  The inclusion problem for \ZCFGr\ is \coNEXP-complete.
\end{theorem}

In the conference version of this paper~\cite{HH14}, we showed that
inclusion is \coNEXP-hard for \ZVASS, even when numbers are encoded in
unary. Our construction was subsequently strengthened in~\cite{HH15}
were it was shown that inclusion is already \coNEXP-hard for
\ZVASS\ when counter updates are restricted to be non-negative and
given in unary. The \coNEXP-lower bound of
Theorem~\ref{thm:inclusion-general} consequently follows
from~\cite{HH15}.

Thanks to our characterization of reachability sets of \ZCFGr\ via
existential Presburger formulas of polynomial size obtained from
Corollary~\ref{cor:rs}, a matching upper bound is also not difficult
to obtain. Let $\mathcal{G}$ and $\mathcal{H}$ be \ZCFGr,
$\vec{u},\vec{v}\in \zint^C$, and let
$\Phi^\mathcal{G}_{\text{rs}}(\vec{x}, \vec{z})$ and
$\Phi^\mathcal{H}_{\text{rs}}(\vec{y}, \vec{z})$ be the formulas from
Corollary~\ref{cor:rs}. We then have that
\begin{align*}
& \mathit{reach}(\mathcal{G}, \vec{u})\subseteq
\mathit{reach}(\mathcal{H}, \vec{v})\\  
\iff & \psi \defeq \neg (\exists \vec{z}.
\Phi^{\mathcal{G}}_\text{rs}(\vec{u}/\vec{x},\vec{z}) \wedge \neg 
(\Phi^{\mathcal{H}}_\text{rs}(\vec{v}/\vec{y},\vec{z})))~\text{is valid.}
\end{align*}
Bringing $\psi$ into prenex normal form yields a $\Pi_2$-PA sentence
for which validity can be decided in $\coNEXP$,
cf.\ Proposition~\ref{prop:pa-complexity}. This concludes the proof
of Theorem~\ref{thm:inclusion-general}.

\subsection{Inclusion for \ZVAS}

In this section, we show that already for \ZVAS, the inclusion problem
is computationally difficult.

\begin{theorem}\label{thm:inclusion-zvas}
  The inclusion problem for \ZVAS\ is
  $\ComplexityFont{\Pi_2^\P}$-complete.
\end{theorem}

In fact, the lower bound already holds when numbers are encoded
in unary (see Theorem~\ref{t:unary} in the following subsection).
The starting point for
our lower bound here (for the binary encoding)
is the following generalization of \textsc{Subset Sum},
which is known to be complete for the second level of the
polynomial hierarchy.  Recall that, unless explicitly stated
otherwise, all numbers in the considered problem settings are written
in binary.
\vspace*{0.25cm}
\problemx{$\Pi_2$-Subset Sum}
{Finite sets $U,V \subseteq \mathbb{N}$ and $t \in \mathbb{N}$.}
{For every $U'\subseteq U$, does there exist a
 $V'\subseteq V$ such that $\sum U' + \sum V' = t$?}
\vspace*{0.25cm}
Here and below, for $A\subseteq \nat$ we use $\sum A$ as a shorthand
for $\sum_{a\in A} a$.
\begin{proposition}[Berman et al.~\cite{BKLPR02}]\label{prop:pi-2-subsetsum}
  \textsc{$\Pi_2$-\textsc{Subset Sum}} is
  $\ComplexityFont{\Pi_2^\P}$-complete.
\end{proposition}
There is no obvious reduction from \textsc{$\Pi_2$-\textsc{Subset Sum}}
to inclusion for \ZVAS. Informally,
the lack of control structure in \ZVAS\ makes it difficult
to encode the alternation of quantifiers (for all $U'$ there exists a $V'$)
and the subset constraints (each element of $U$, respectively $V$, participates at most
once in $U'$, respectively $V'$).
Accordingly, we define another variant of
\textsc{Subset Sum}, implicit in~\cite{CM14}.
\vspace*{0.25cm}
\problemx{Simultaneous Subset Sum}
{A finite set $W\subseteq \mathbb{N}$, and $h, 2^m, t\in \mathbb{N}$ such that $t<h$.}
{For every $i\in [0,2^m-1]$, does there exist a $W'\subseteq W$ such that
$\sum W' = t + i\cdot h$?}
\vspace*{0.25cm}
\begin{lemma}\label{lem:simultaneous-ss-complexity}
  \textsc{Simultaneous Subset Sum} is
  $\ComplexityFont{\Pi_2^\P}$-complete.
\end{lemma}
\begin{proof}
  The problem is easily seen to be in $\ComplexityFont{\Pi_2^\P}$. To
  show hardness, let $U=\{u_1,\ldots, u_r\}, V=\{v_1,\ldots,v_s
  \}\subseteq \mathbb{N}$ and $t\in \mathbb{N}$ form an instance of
  \textsc{$\Pi_2$-\textsc{Subset Sum}}. We define the corresponding
  instance of \textsc{Simultaneous Subset Sum} as follows:
  \begin{itemize}
  \item $h\defeq \sum U + \sum V + 1$;
  \item $W \defeq \{ u_1 + h, u_2 + 2\cdot h, \ldots, u_r + 2^{r - 1}
    \cdot h, v_1, \ldots, v_s\}$;
  \item $m\defeq r$, and $t$ is unchanged.
  \end{itemize}
  We now show that this reduction is faithful. With no loss of
  generality, we may assume $t<h$, otherwise the original instance is
  clearly a no-instance. We actually show a slightly stronger
  statement: define a bijection between $T\defeq \{ t + i\cdot h
  \colon i\in [0,2^r-1]\}$ and $2^U$ as follows: associate with $t +
  i\cdot h\in T$ the set $U_i\subseteq U$ such that
\[
u_j \in U_i \iff \text{$2^{j-1}$ has non-zero coefficient in
  the binary expansion of $i$,}
\]
i.e., $U_i$ is such that $i = \sum_{u_j\in U_i} 2^{j-1}$. We claim that every $t
+ i\cdot h\in T$ can be represented as a sum of some $W'\subseteq W$
if and only if for the subset $U_i\subseteq U$ there is some
$V'\subseteq V$ such that $\sum U_i + \sum V' = t$. Indeed, observe
that
  \begin{align}
    & \sum U_i + \sum V' = t \notag\\
    \iff & \sum_{u_j\in U_i} u_j + \sum V' = t \notag\\
    \iff & \sum_{u_j\in U_i} u_j + \sum V' + \sum_{u_j\in U_i} 2^{j-1}\cdot h 
    = t + i\cdot h \notag \\
    \label{eqn:simultaneous-ss-1}
    \iff & \sum_{u_j\in U_i} (u_j + 2^{j-1} \cdot h) + \sum V' = t + i\cdot h\\
    \label{eqn:simultaneous-ss-2}
    \iff & \sum W' = t + i \cdot h \quad \text{for some } W'\subseteq W,
  \end{align}
  where the implication $\eqref{eqn:simultaneous-ss-2}\Rightarrow
  \eqref{eqn:simultaneous-ss-1}$ holds by our choice of $h$.
  \qed
\end{proof}
We now apply Lemma~\ref{lem:simultaneous-ss-complexity} in order to
obtain the lower bound for \ZVAS\ inclusion. Let $W=\{ w_1, \ldots,
w_n\} \subseteq \nat$ and $h, 2^m, t\in \nat$ define an instance of
\textsc{Simultaneous Subset Sum}. Set $\vec{w} \defeq
(w_1,\ldots,w_n)$, then this instance is a yes-instance if and only if
\begin{align}\label{eqn:simultaneous-ss-vectors}
  \text{for all } i \in [0,2^m-1] \text{ there exists a }
  \vec{y}\in \{0,1\}^n \text{ such that } \vec{w} \cdot \vec{y} = t + i\cdot h.
\end{align}
It follows from the discussion in Section~\ref{sec:def},
p.~\pageref{anc:matrix-reachability}, that the \ZVAS\ inclusion
problem can equivalently be expressed as follows: For matrices $A \in
\zint^{d \times r}$ and $B \in \zint^{d\times s}$, and some
$\vec{v}\in \zint^d$, decide whether
\begin{align}\label{eqn:zvas-inclusion-vectors}
  \text{for all } \vec{x} \in \nat^r, \text{there exists a } 
  \vec{y} \in \nat^s \text{ such that } A\cdot \vec{x} + B\cdot \vec{y} =
  \vec{v}.
\end{align}
We now transform~\eqref{eqn:simultaneous-ss-vectors} into the
form~\eqref{eqn:zvas-inclusion-vectors}, thus proving
$\ComplexityFont{\Pi_2^\P}$-hardness of \ZVAS\ inclusion. Observe
that~\eqref{eqn:simultaneous-ss-vectors} is almost of the same form
as~\eqref{eqn:zvas-inclusion-vectors}. However, the domains of the
quantified variables in~\eqref{eqn:simultaneous-ss-vectors}
and~\eqref{eqn:zvas-inclusion-vectors} disagree. In order to overcome
this issue, first we observe that the existence of some $\vec{y}\in
\{0,1\}^n$ is equivalent to the existence of $\vec{y},\vec{z}\in
\nat^n$ such that $\vec{y} + \vec{z} = \vec{1}$. Second, the
restriction of $i$ to numbers less than $2^m$ can be avoided by
introducing another existentially quantified variable $c \in \nat$ and
replacing $i$ with $i - 2^m\cdot c$
in~\eqref{eqn:simultaneous-ss-vectors}. Informally speaking, this
ensures that $i$ is evaluated only modulo $2^m$ and, effectively, does
not ``overflow''. Putting everything together, we claim
that~\eqref{eqn:simultaneous-ss-vectors} is equivalent to the
following condition:
\begin{multline}\label{eqn:simultaneous-ss-vectors-transformed}
  \text{for all } i\in \nat \text{ there exist } \vec{y},\vec{z}\in \nat^n
  \text{ and } c\in \nat \text{ such that}\\
  \vec{w} \cdot \vec{y} = t + (i - 2^m\cdot c) \cdot h \text{ and }
  \vec{y} + \vec{z} = \vec{1}.
\end{multline}
Indeed, \eqref{eqn:simultaneous-ss-vectors}
implies~\eqref{eqn:simultaneous-ss-vectors-transformed};
conversely, \eqref{eqn:simultaneous-ss-vectors-transformed}
    implies~\eqref{eqn:simultaneous-ss-vectors}, because for
$i<2^m$ no $c> 0$ can satisfy the first equation
in~\eqref{eqn:simultaneous-ss-vectors-transformed}:
the right-hand side is $t+(i-2^m \cdot c)\cdot h \le t
- h< 0$, while the left-hand side is $\vec{w}\cdot \vec{y} \ge 0$ for all $\vec{y}\in
\nat^n$. It is readily verified
that~\eqref{eqn:simultaneous-ss-vectors-transformed} is of the
form~\eqref{eqn:zvas-inclusion-vectors} with $r=1$, $s=2 \cdot n + 1$,
$d=n+1$ and
\begin{gather*}
  A \defeq \begin{pmatrix}-h \\ \vec{0} \end{pmatrix}, \quad B \defeq
  \begin{pmatrix}
    \vec{w}^\intercal & \vec{0}^\intercal & 2^m\cdot h\\
    I_n & I_n & \vec{0}
  \end{pmatrix},\quad
  \vec{v} \defeq \begin{pmatrix} t\\ \vec{1} \end{pmatrix}.
\end{gather*}
This concludes the proof of the $\ComplexityFont{\Pi_2^\P}$-hardness of
\ZVAS\ inclusion when numbers are encoded in binary.

We now turn towards a matching upper bound for \ZVAS\ inclusion. Given
a \ZVAS\ $\mathcal{G}=(\{S\},C,P,S)$ such that $P=\{ (S,\vec{v}_1,S),
\ldots, (S,\vec{v}_n,S) \}$ and a configuration $(S,\vec{v})$, we
obviously have
\[
\mathit{reach}(\mathcal{G},\vec{v}) = \left\{ \vec{v} + \sum_{1\le
  i\le n} \lambda_i \cdot \vec{v}_i : \lambda_i\in \nat, 1\le i\le n
\right\},
\]
which is a linear (and thus semi-linear) set in $\zint^d$.
It follows from the results
in~\cite{CH16} (and implicitly also from~\cite{Huy86}) that the
inclusion problem for semi-linear sets in~$\zint^d$ given by their
generators is in $\ComplexityFont{\Pi_2^\P}$ (and is, in fact,
$\ComplexityFont{\Pi_2^\P}$-complete). As a consequence, we
conclude that the inclusion problem for \ZVAS\ is also contained in
$\ComplexityFont{\Pi_2^\P}$, and hence is
$\ComplexityFont{\Pi_2^\P}$-complete.

%


\subsection{Inclusion for \ZVAS\ under unary encoding of integers}

It is interesting to note that, modulo standard computational complexity
assumptions, both $\Pi_2$-\textsc{Subset Sum} and
\textsc{Simultaneous Subset Sum} are only
$\ComplexityFont{\Pi_2^\P}$-hard if numbers are represented in binary:
it is folklore that \textsc{Subset Sum} has a pseudo--polynomial time
dynamic programming algorithm, and thus the unary versions
of $\Pi_2$-\textsc{Subset Sum} and \textsc{Simultaneous Subset Sum}
are not $\ComplexityFont{\Pi_2^\P}$-hard unless the polynomial hierarchy
collapses.
This phenomenon, however, does not extend to the inclusion problem for
\ZVAS.
Specifically, Theorem~\ref{t:unary} gives a stronger form of the lower
bound in Theorem~\ref{thm:inclusion-zvas}:

\begin{theorem}
\label{t:unary}
The inclusion problem for \ZVAS\ remains
$\ComplexityFont{\Pi_2^\P}$-hard even when numbers are encoded in
unary.
\end{theorem}
\begin{proof}
  We reduce the inclusion problem with numbers encoded in binary to
  the inclusion problem with numbers encoded in unary. More precisely,
  we transform the problem described in
  Equation~(\ref{eqn:zvas-inclusion-vectors}) into an instance of the
  same problem with numbers encoded in unary. Recall that this problem
  is to decide, given matrices $A \in \zint^{d\times r}$, $B \in
  \zint^{d\times s}$ and a vector $\vec{v} \in \zint^d$ whether
  \begin{align*}
    \text{for all } \vec{x} \in \nat^r, \text{there exists } 
    \vec{y} \in \nat^s \text{ such that } A\cdot \vec{x} + B\cdot \vec{y} =
    \vec{v}.
  \end{align*}
  We construct an instance $(A', B', \vec{v}')$ of the same problem
  but where $A', B'$ and $\vec{v}$ use only numbers among
  $\{-2,-1,0,1,2 \}$. We first introduce some auxiliary definitions.

  Let $m$ be the minimal number of bits sufficient to write in binary every
  number occurring in $A$, $B$ and $\vec{v}$ disregarding the signs---%
  in other words, $m$ is the smallest
  natural number such that the absolute value of every entry of $A$
  (respectively $B$ and $\vec{v}$) is smaller than $2^m$. Given a
  vector $\vec{x} = (x_1, \dots, x_d) \in [-2^m+1,2^m-1]^d$, define
  $b(\vec{x}) \in \{-1, 0,1\}^{d\cdot m}$, the \emph{binary expansion}
  of $\vec{x}$, as
  \begin{align*}
    b(\vec{x}) = (x_1^{(m-1)}, x_1^{(m-2)}, \dots, x_1^{(0)}, x_2^{(m-1)}, \dots,
    x_d^{(0)})
  \end{align*}
  where, for $i \in [d]$, $x_i = \sum_{j=0}^{m-1} x_i^{(j)} \cdot 2^j$
  denotes the unique binary representation of $x_i$ (for negative
  numbers, coefficients $x_i^{(j)}$ are in $\{-1,0\}$, and for positive numbers
  they are in $\{ 0,1 \}$). Conversely, given
  \begin{equation*}
  \vec{y} = (y_1^{(m-1)},
  y_1^{(m-2)}, \dots, y_1^{(0)}, y_2^{(m-1)}, \dots, y_d^{(0)}) \in
  \zint^{d\cdot m},
  \end{equation*}
  we define the reverse function $r(\vec{y})
  = (r_1, \dots, r_d) \in \zint^d$ with $r_i = \sum_{j=0}^{m-1}
  y_i^{(j)} \cdot 2^j$ for $i \in [d]$. Note that for any $\vec{x}\in
  [-2^m+1, 2^m-1]^d$, $r(b(\vec{x})) = \vec{x}$, but $b(r(\vec{y}))$
  is not necessarily equal to $\vec{y}$ since components of $\vec{y}
  \in \zint^{d\cdot m}$ do not have to belong to $\{ 0,1\}$ or $\{-1,
  0\}$. 

  Let us also introduce the following definition.
  Given an integer $n\in \zint$, its \emph{weak binary
    representation}
  is an expansion of~$n$ as a sum of powers of~$2$ with arbitrary
  coefficients from~$\zint$ (the
  usual binary representation only allows coefficients~$0$ and~$1$).
  Then, with $d=1$, the set $\{ \vec{y} \in
  \zint^m : r(\vec{y}) = n \}$ is the set of all weak binary
  representations of $n$ of \emph{height~$m$}. For instance, $(0,0,1,1,1)$
  and $(1,0,-1,-3,1)$ are two weak binary representations of $7$, both of height~$5$.

  Now define the matrix $D_m = (d_{i,j})\in [-2,2]^{m\times
    (m-1)}$ by the following rule:
\[
d_{i,j} \defeq 
\begin{cases}
   1 & \text{if } i=j, \\
   -2 & \text{if } i=j+1, \text{ and} \\
   0 & \text{otherwise}.
\end{cases}
\]
For instance, 
\begin{align*}
D_5 =
\begin{pmatrix}
1 & 0 & 0 & 0 \\
-2 & 1 & 0 & 0 \\
0 & -2 & 1 & 0 \\
0 & 0 & -2 & 1 \\
0 & 0 & 0 & -2 
\end{pmatrix}
.
\end{align*}
The image $\{ D_m \cdot\vec{z} \ : \ \vec{z} \in \zint^{m-1} \}$ of
the matrix $D_m$ is the set of all weak binary representations of
height $m$ of the integer $0$:
\begin{claim}
\label{claim:binary-property}
For every vector $\vec{y} \in \zint^m$,
$r(\vec{y}) = 0$ iff
there exists a $\vec{z} \in \zint^{m-1}$ such that $D_m\cdot \vec{z} = \vec{y}$.
\end{claim}
We prove Claim~\ref{claim:binary-property} at the end of the section.
By linearity, it follows that the set $\{ \vec{y} + D_m\cdot \vec{z} : \vec{z} \in \zint^m \}$ is the set of all weak binary representations of $r(\vec{y})$. For instance, 
\begin{align*}
\begin{pmatrix}
1\\ 0 \\ -1 \\ -3 \\ 1
\end{pmatrix} 
+ D_5 \cdot 
\begin{pmatrix} -1 \\ -2 \\ -2 \\ 0 \end{pmatrix}
 = 
\begin{pmatrix}
0 \\ 0 \\ 1 \\ 1 \\ 1
\end{pmatrix}
,
\end{align*}
and all height-$5$ weak binary representations of $7$ can be
obtained this way. The same property can be obtained for vectors of
dimension $d$ by defining
\begin{align*}
E_m^d =
\begin{pmatrix}
D_m & 0 & \dots & 0 \\
0 & D_m & \dots & 0 \\
\vdots & \vdots & \ddots & \vdots\\
0 & 0 & \dots & D_m 
\end{pmatrix} \in [-2,2]^{(d\cdot m) \times (d \cdot (m-1))}.
\end{align*}

We now define the following instance of \ZVAS\ inclusion:
\begin{itemize}
\item $A' = b(A) \in \{-1, 0,1\}^{(d\cdot m) \times r}$, i.e., $A'$ is
  the matrix whose columns are the vectors $b(\vec{a})$ for every
  column $\vec{a}$ of $A$;
\item $B' = 
\begin{pmatrix}
b(B) & E_m^d & -E_m^d
\end{pmatrix}
\in [-2,2]^{(d\cdot m) \times (s+ 2d\cdot (m-1))}$; and
\item $\vec{v}' = b(\vec{v}) \in \{-1, 0, 1\}^{d \cdot m}$.
\end{itemize}

This instance is a yes-instance if and only if for all $\vec{x} \in
\nat^r$, there exists a $(\vec{y},\vec{z}_1,\vec{z}_2) \in \nat^{s+
  2d\cdot (m-1)}$ such that
\begin{align*}
b(A)\cdot \vec{x} + 
\begin{pmatrix}
b(B) & E_m^d & -E_m^d
\end{pmatrix}
\cdot
\begin{pmatrix}
\vec{y} \\
\vec{z}_1 \\
\vec{z}_2
\end{pmatrix}
= b(\vec{v}),
\end{align*}
i.e., if there exists a $\vec{t} \in \zint^{d\cdot (m-1)}$ such that
\begin{align*}
b(B)\cdot \vec{y} + E_m^d \cdot \vec{t} = b(\vec{v}) - b(A)\cdot\vec{x}.
\end{align*}
According to Claim~\ref{claim:binary-property}, there exists such a
$\vec{t}$ if and only if $b(B)\cdot\vec{y}$ and $b(\vec{v}) -
b(A)\cdot\vec{x}$ are two weak binary representations of the same
vector. By application of $r$ on both sides and due to the linearity of
$r$, we obtain that for any $\vec x\in \nat^r$ there exists some $\vec
y\in \nat^s$ such that $A\cdot \vec x + B \cdot \vec y = \vec v$.
This completes the proof of Theorem~\ref{t:unary}.
\qed

\end{proof}


It remains to prove Claim~\ref{claim:binary-property}. Recall that we
wish to show that for any $\vec{y} = (y^{(m-1)}, \dots, y^{(0)})$,
\begin{align*}
r(\vec{y}) \defeq \sum_{j=0}^{m-1} y^{(j)} \cdot 2^j = 0 \iff \exists 
\vec{z} .\ D_m \cdot \vec{z} = \vec{y}.
\end{align*}
Let $C_m$ be the square matrix that consists of the $m-1$ first rows of
$D_m$, i.e., $C_m$ is such that
\[
D_m =
\begin{pmatrix}
C_m  \\
0\cdots0 \ {-2}
\end{pmatrix}.
\]
We then have:
\begin{align}
\label{eqn:img}
  \exists \vec{z} .\ D_m \cdot \vec{z} = \vec{y} \iff  
  \exists \vec{z} .\ C_m \cdot \vec{z} = 
  (y^{(m-1)},\dots, y^{(1)}) \text{ and } {-2}\cdot z^{(1)} = y^{(0)}.
\end{align}
Since $C_m$ is invertible, for every $\vec{y}$ there exists exactly
one $\vec{z} = (z^{(m-1)}, \dots, z^{(1)})$ such that $C_m \cdot
\vec{z} = (y^{(m-1)},\dots, y^{(1)})$; in particular $z^{(i)} =
\sum_{j=i}^{m-1} y^{(j)} \cdot 2^{j-i}$ for every $i \in [m-1]$.  Note
that $\vec{z} \in \zint^{m-1}$ whenever $\vec{y} \in \zint^{m-1}$.
Therefore, the equivalence in~\eqref{eqn:img} can be reformulated
and continued as follows:
\begin{align*}
 \exists \vec{z} .\ D_m \cdot \vec{z} = \vec{y} 
&\iff -2 \sum_{j=1}^{m-1} y^{(j)}\cdot 2^{j-1} = y^{(0)} \\
&\iff 0 = \sum_{j=1}^{m-1} y^{(j)}\cdot 2^{j} + y^{(0)}\cdot 2^0 \\
&\iff r(\vec{y}) = 0.
\end{align*}
This concludes the proof of Claim~\ref{claim:binary-property}.




\section{Conclusion}
In this paper, we studied standard decision problems for \ZCFGr, an
extension of context-free commutative grammars with integer counters
and reset operations on them. We showed that reachability and
coverability are logarithmic-space inter-reducible in this class and
\NP-complete. For our \NP-upper bound, we showed that the reachability
relation for \ZCFGr\ can be defined by an existential formula of
Presburger arithmetic of polynomial size. In particular, this implies
that \ZCFGr\ have semi-linear reachability sets. Moreover, we showed
that inclusion for \ZCFGr \ is, in general, \coNEXP-complete, and
$\ComplexityFont{\Pi_2^\P}$-complete for \ZVAS, a subclass of \ZCFGr.
In order to show the latter lower bound, we introduced a new
$\ComplexityFont{\Pi_2^\P}$-complete decision problem
\textsc{Simultaneous Subset Sum}, a variant of the classical
\textsc{Subset Sum} problem.

One can view \ZCFGr\ as an over-approximation of classical reset Petri
nets in which places may contain a negative number of tokens. Hence,
Theorem~\ref{cor:reachability-presburger} enables witnessing
non-reachability in reset Petri nets in \coNP, i.e., at comparatively low
computational costs given that the problem is, in general,
undecidable. In particular, our characterization of reachability in
terms of existential Presburger arithmetic immediately enables the
use of SMT~solvers and thus paves the way for an easy implementation of our
approach. This approach, over-approximating reachability in Petri
nets, has recently been proved surprisingly efficient when applied to
real-world instances~\cite{BFHH15}. As for future work, it would be
interesting to investigate whether \ZCFGr\ can be extended with
transfer operations while retaining definability of their reachability
sets in Presburger arithmetic.

\subsubsection*{Acknowledgments.} 
We would like to thank Sylvain Schmitz, Philippe Schnoebelen and the
anonymous reviewers of RP'14 for their helpful comments and
suggestions on an earlier version of this paper.





\bibliographystyle{elsarticle-num-names}
\bibliography{bibliography}

\begin{thebibliography}{49}
\providecommand{\natexlab}[1]{#1}
\providecommand{\url}[1]{\texttt{#1}}
\providecommand{\urlprefix}{URL }
\expandafter\ifx\csname urlstyle\endcsname\relax
  \providecommand{\doi}[1]{doi:\discretionary{}{}{}#1}\else
  \providecommand{\doi}[1]{doi:\discretionary{}{}{}\begingroup
  \urlstyle{rm}\url{#1}\endgroup}\fi
\providecommand{\bibinfo}[2]{#2}

\bibitem[{German and Sistla(1992)}]{GS92}
\bibinfo{author}{S.~M. German}, \bibinfo{author}{A.~P. Sistla},
  \bibinfo{title}{Reasoning about Systems with Many Processes},
  \bibinfo{journal}{J. {ACM}} \bibinfo{volume}{39}~(\bibinfo{number}{3})
  (\bibinfo{year}{1992}) \bibinfo{pages}{675--735}.

\bibitem[{Lipton(1976)}]{Lipt76}
\bibinfo{author}{R.~Lipton}, \bibinfo{title}{The Reachability Problem is
  Exponential-Space-Hard}, \bibinfo{type}{Tech. Rep.},
  \bibinfo{institution}{Yale University}, \bibinfo{address}{New Haven, CT},
  \bibinfo{year}{1976}.

\bibitem[{Rackoff(1978)}]{Rack78}
\bibinfo{author}{C.~Rackoff}, \bibinfo{title}{The covering and boundedness
  problems for vector addition systems}, \bibinfo{journal}{Theor. Comput. Sci.}
  \bibinfo{volume}{6}~(\bibinfo{number}{2}) (\bibinfo{year}{1978})
  \bibinfo{pages}{223--231}.

\bibitem[{Mayr(1984)}]{Mayr81}
\bibinfo{author}{E.~W. Mayr}, \bibinfo{title}{An Algorithm for the General
  Petri Net Reachability Problem}, \bibinfo{journal}{{SIAM} J. Comput.}
  \bibinfo{volume}{13}~(\bibinfo{number}{3}) (\bibinfo{year}{1984})
  \bibinfo{pages}{441--460}.

\bibitem[{Kosaraju(1982)}]{Kos82}
\bibinfo{author}{S.~R. Kosaraju}, \bibinfo{title}{Decidability of Reachability
  in Vector Addition Systems (Preliminary Version)}, in: \bibinfo{editor}{H.~R.
  Lewis}, \bibinfo{editor}{B.~B. Simons}, \bibinfo{editor}{W.~A. Burkhard},
  \bibinfo{editor}{L.~H. Landweber} (Eds.), \bibinfo{booktitle}{Symposium on
  Theory of Computing ({STOC}'82)}, \bibinfo{publisher}{{ACM}},
  \bibinfo{pages}{267--281}, \bibinfo{year}{1982}.

\bibitem[{Lambert(1992)}]{Lam92}
\bibinfo{author}{J.~Lambert}, \bibinfo{title}{A Structure to Decide
  Reachability in Petri Nets}, \bibinfo{journal}{Theor. Comput. Sci.}
  \bibinfo{volume}{99}~(\bibinfo{number}{1}) (\bibinfo{year}{1992})
  \bibinfo{pages}{79--104}.

\bibitem[{Leroux(2012)}]{Leroux12}
\bibinfo{author}{J.~Leroux}, \bibinfo{title}{Vector Addition Systems
  Reachability Problem {(A} Simpler Solution)}, in:
  \bibinfo{editor}{A.~Voronkov} (Ed.), \bibinfo{booktitle}{Turing-100 - The
  Alan Turing Centenary}, vol.~\bibinfo{volume}{10} of
  \emph{\bibinfo{series}{EPiC Series}}, \bibinfo{publisher}{EasyChair},
  \bibinfo{pages}{214--228}, \bibinfo{year}{2012}.

\bibitem[{Leroux and Schmitz(2015)}]{LS15}
\bibinfo{author}{J.~Leroux}, \bibinfo{author}{S.~Schmitz},
  \bibinfo{title}{Demystifying Reachability in Vector Addition Systems}, in:
  \bibinfo{booktitle}{Logic in Computer Science ({LICS}'15)},
  \bibinfo{publisher}{{IEEE} Computer Society}, \bibinfo{pages}{56--67},
  \bibinfo{year}{2015}.

\bibitem[{Hack(1976)}]{Hack76}
\bibinfo{author}{M.~Hack}, \bibinfo{title}{The equality problem for vector
  addition systems is undecidable}, \bibinfo{journal}{Theor. Comput. Sci.}
  \bibinfo{volume}{2}~(\bibinfo{number}{1}) (\bibinfo{year}{1976})
  \bibinfo{pages}{77--95}.

\bibitem[{Jan{\v c}ar(2001)}]{Jan01}
\bibinfo{author}{P.~Jan{\v c}ar}, \bibinfo{title}{Nonprimitive recursive
  complexity and undecidability for {P}etri net equivalences},
  \bibinfo{journal}{Theor. Comput. Sci.}
  \bibinfo{volume}{256}~(\bibinfo{number}{1--2}) (\bibinfo{year}{2001})
  \bibinfo{pages}{23--30}.

\bibitem[{Kaiser et~al.(2014)Kaiser, Kroening, and Wahl}]{KKW14}
\bibinfo{author}{A.~Kaiser}, \bibinfo{author}{D.~Kroening},
  \bibinfo{author}{T.~Wahl}, \bibinfo{title}{A Widening Approach to
  Multithreaded Program Verification}, \bibinfo{journal}{{ACM} Trans. Program.
  Lang. Syst.} \bibinfo{volume}{36}~(\bibinfo{number}{4})
  (\bibinfo{year}{2014}) \bibinfo{pages}{14:1--14:29}.

\bibitem[{Wynn et~al.(2009)Wynn, van~der Aalst, ter Hofstede, and
  Edmond}]{WAHE09}
\bibinfo{author}{M.~T. Wynn}, \bibinfo{author}{W.~M.~P. van~der Aalst},
  \bibinfo{author}{A.~H.~M. ter Hofstede}, \bibinfo{author}{D.~Edmond},
  \bibinfo{title}{Synchronization and Cancelation in Workflows Based on Reset
  Nets}, \bibinfo{journal}{Int. J. Cooperative Inf. Syst.}
  \bibinfo{volume}{18}~(\bibinfo{number}{1}) (\bibinfo{year}{2009})
  \bibinfo{pages}{63--114}.

\bibitem[{Dufourd et~al.(1998)Dufourd, Finkel, and Schnoebelen}]{DFS98}
\bibinfo{author}{C.~Dufourd}, \bibinfo{author}{A.~Finkel},
  \bibinfo{author}{P.~Schnoebelen}, \bibinfo{title}{Reset Nets Between
  Decidability and Undecidability}, in: \bibinfo{editor}{K.~G. Larsen},
  \bibinfo{editor}{S.~Skyum}, \bibinfo{editor}{G.~Winskel} (Eds.),
  \bibinfo{booktitle}{Automata, Languages and Programming ({ICALP}'98)}, vol.
  \bibinfo{volume}{1443} of \emph{\bibinfo{series}{Lect. Notes Comp. Sci.}},
  \bibinfo{publisher}{Springer}, \bibinfo{pages}{103--115},
  \bibinfo{year}{1998}.

\bibitem[{Finkel et~al.(2013)Finkel, G{\"{o}}ller, and Haase}]{FGH13}
\bibinfo{author}{A.~Finkel}, \bibinfo{author}{S.~G{\"{o}}ller},
  \bibinfo{author}{C.~Haase}, \bibinfo{title}{Reachability in Register Machines
  with Polynomial Updates}, in:  \cite{DBLP:conf/mfcs/2013},
  \bibinfo{pages}{409--420}, \bibinfo{year}{2013}.

\bibitem[{Schnoebelen(2010)}]{Schn10}
\bibinfo{author}{P.~Schnoebelen}, \bibinfo{title}{Revisiting Ackermann-Hardness
  for Lossy Counter Machines and Reset Petri Nets}, in:
  \bibinfo{editor}{P.~Hlinen{\'{y}}}, \bibinfo{editor}{A.~Kucera} (Eds.),
  \bibinfo{booktitle}{Mathematical Foundations of Computer Science
  ({MFCS}'10)}, vol. \bibinfo{volume}{6281} of \emph{\bibinfo{series}{Lect.
  Notes Comp. Sci.}}, \bibinfo{publisher}{Springer}, \bibinfo{pages}{616--628},
  \bibinfo{year}{2010}.

\bibitem[{Huynh(1985)}]{Huy85}
\bibinfo{author}{D.~Huynh}, \bibinfo{title}{The complexity of equivalence
  problems for commutative grammars}, \bibinfo{journal}{Inform. Control}
  \bibinfo{volume}{66}~(\bibinfo{number}{1--2}) (\bibinfo{year}{1985})
  \bibinfo{pages}{103--121}.

\bibitem[{Yen(1997)}]{Yen97}
\bibinfo{author}{H.~Yen}, \bibinfo{title}{On Reachability Equivalence for
  BPP-Nets}, \bibinfo{journal}{Theor. Comput. Sci.}
  \bibinfo{volume}{179}~(\bibinfo{number}{1--2}) (\bibinfo{year}{1997})
  \bibinfo{pages}{301--317}.

\bibitem[{Mayr and Weihmann(2015)}]{MW15}
\bibinfo{author}{E.~W. Mayr}, \bibinfo{author}{J.~Weihmann},
  \bibinfo{title}{Complexity Results for Problems of Communication-Free {P}etri
  Nets and Related Formalisms}, \bibinfo{journal}{Fundam. Inform.}
  \bibinfo{volume}{137}~(\bibinfo{number}{1}) (\bibinfo{year}{2015})
  \bibinfo{pages}{61--86}.

\bibitem[{Hopcroft and Pansiot(1979)}]{HP79}
\bibinfo{author}{J.~E. Hopcroft}, \bibinfo{author}{J.~Pansiot},
  \bibinfo{title}{On the Reachability Problem for 5-Dimensional Vector Addition
  Systems}, \bibinfo{journal}{Theor. Comput. Sci.} \bibinfo{volume}{8}
  (\bibinfo{year}{1979}) \bibinfo{pages}{135--159}.

\bibitem[{Huynh(1983)}]{Huy83}
\bibinfo{author}{D.~T. Huynh}, \bibinfo{title}{Commutative Grammars: The
  Complexity of Uniform Word Problems}, \bibinfo{journal}{Inform. Control}
  \bibinfo{volume}{57}~(\bibinfo{number}{1}) (\bibinfo{year}{1983})
  \bibinfo{pages}{21--39}.

\bibitem[{Huynh(1986)}]{Huy86}
\bibinfo{author}{D.~T. Huynh}, \bibinfo{title}{A Simple Proof for the
  ${\Sigma}^p_2$ Upper Bound of the Inequivalence Problem for Semilinear Sets},
  \bibinfo{journal}{Elektron. Inform. Kybernet.}
  \bibinfo{volume}{22}~(\bibinfo{number}{4}) (\bibinfo{year}{1986})
  \bibinfo{pages}{147--156}.

\bibitem[{Esparza(1997)}]{Esp97}
\bibinfo{author}{J.~Esparza}, \bibinfo{title}{Petri Nets, Commutative
  Context-Free Grammars, and Basic Parallel Processes},
  \bibinfo{journal}{Fundam. Inform.} \bibinfo{volume}{31}~(\bibinfo{number}{1})
  (\bibinfo{year}{1997}) \bibinfo{pages}{13--25}.

\bibitem[{Kopczynski and To(2010)}]{KT10}
\bibinfo{author}{E.~Kopczynski}, \bibinfo{author}{A.~W. To},
  \bibinfo{title}{Parikh Images of Grammars: Complexity and Applications}, in:
  \bibinfo{booktitle}{Logic in Computer ({LICS}'10)},
  \bibinfo{publisher}{{IEEE}}, \bibinfo{pages}{80--89}, \bibinfo{year}{2010}.

\bibitem[{Kopczy{\'n}ski(2015)}]{Kop15}
\bibinfo{author}{E.~Kopczy{\'n}ski}, \bibinfo{title}{Complexity of Problems of
  Commutative Grammars}, \bibinfo{journal}{Log. Meth. Comput. Sci.}
  \bibinfo{volume}{11}~(\bibinfo{number}{1}).

\bibitem[{Haase and Hofman(2016)}]{HH15}
\bibinfo{author}{C.~Haase}, \bibinfo{author}{P.~Hofman},
  \bibinfo{title}{Tightening the Complexity of Equivalence Problems for
  Commutative Grammars}, in: \bibinfo{editor}{N.~Ollinger},
  \bibinfo{editor}{H.~Vollmer} (Eds.), \bibinfo{booktitle}{Symposium on
  Theoretical Aspects of Computer Science ({STACS'15})},
  vol.~\bibinfo{volume}{47} of \emph{\bibinfo{series}{LIPIcs}},
  \bibinfo{publisher}{Schloss Dagstuhl - Leibniz-Zentrum fuer Informatik},
  \bibinfo{pages}{41:1--41:14}, \bibinfo{year}{2016}.

\bibitem[{Haase and Halfon(2014)}]{HH14}
\bibinfo{author}{C.~Haase}, \bibinfo{author}{S.~Halfon},
  \bibinfo{title}{Integer Vector Addition Systems with States}, in:
  \bibinfo{editor}{J.~Ouaknine}, \bibinfo{editor}{I.~Potapov},
  \bibinfo{editor}{J.~Worrell} (Eds.), \bibinfo{booktitle}{Reachability
  Problems ({RP}'14)}, vol. \bibinfo{volume}{8762} of
  \emph{\bibinfo{series}{Lect. Notes Comp. Sci.}},
  \bibinfo{publisher}{Springer}, \bibinfo{pages}{112--124},
  \bibinfo{year}{2014}.

\bibitem[{Mayr and Weihmann(2016)}]{MW13}
\bibinfo{author}{E.~W. Mayr}, \bibinfo{author}{J.~Weihmann},
  \bibinfo{title}{Completeness Results for Generalized Communication-free Petri
  Nets with Arbitrary Arc Multiplicities}, \bibinfo{journal}{Fundam. Inform.}
  \bibinfo{volume}{143}~(\bibinfo{number}{3--4}) (\bibinfo{year}{2016})
  \bibinfo{pages}{355--391}.

\bibitem[{Plandowski and Rytter(1999)}]{PR99}
\bibinfo{author}{W.~Plandowski}, \bibinfo{author}{W.~Rytter},
  \bibinfo{title}{Complexity of Language Recognition Problems for Compressed
  Words}, in: \bibinfo{editor}{J.~Karhum{\"a}ki}, \bibinfo{editor}{H.~Maurer},
  \bibinfo{editor}{G.~P\u{a}un}, \bibinfo{editor}{G.~Rozenberg} (Eds.),
  \bibinfo{booktitle}{Jewels are Forever}, \bibinfo{pages}{262--272},
  \bibinfo{year}{1999}.

\bibitem[{Seidl et~al.(2004)Seidl, Schwentick, Muscholl, and
  Habermehl}]{SSMH04}
\bibinfo{author}{H.~Seidl}, \bibinfo{author}{T.~Schwentick},
  \bibinfo{author}{A.~Muscholl}, \bibinfo{author}{P.~Habermehl},
  \bibinfo{title}{Counting in Trees for Free}, in:
  \bibinfo{editor}{J.~D{\'{\i}}az}, \bibinfo{editor}{J.~Karhum{\"{a}}ki},
  \bibinfo{editor}{A.~Lepist{\"{o}}}, \bibinfo{editor}{D.~Sannella} (Eds.),
  \bibinfo{booktitle}{Automata, Languages and Programming ({ICALP}'04)}, vol.
  \bibinfo{volume}{3142} of \emph{\bibinfo{series}{Lect. Notes Comp. Sci.}},
  \bibinfo{publisher}{Springer}, \bibinfo{pages}{1136--1149},
  \bibinfo{year}{2004}.

\bibitem[{Verma et~al.(2005)Verma, Seidl, and Schwentick}]{VSS05}
\bibinfo{author}{K.~N. Verma}, \bibinfo{author}{H.~Seidl},
  \bibinfo{author}{T.~Schwentick}, \bibinfo{title}{On the Complexity of
  Equational Horn Clauses}, in: \bibinfo{editor}{R.~Nieuwenhuis} (Ed.),
  \bibinfo{booktitle}{Automated Deduction - CADE-20}, vol.
  \bibinfo{volume}{3632} of \emph{\bibinfo{series}{Lect. Notes Comp. Sci.}},
  \bibinfo{publisher}{Springer}, \bibinfo{pages}{337--352},
  \bibinfo{year}{2005}.

\bibitem[{Haase et~al.(2009)Haase, Kreutzer, Ouaknine, and Worrell}]{HKOW09}
\bibinfo{author}{C.~Haase}, \bibinfo{author}{S.~Kreutzer},
  \bibinfo{author}{J.~Ouaknine}, \bibinfo{author}{J.~Worrell},
  \bibinfo{title}{Reachability in Succinct and Parametric One-Counter
  Automata}, in: \bibinfo{editor}{M.~Bravetti}, \bibinfo{editor}{G.~Zavattaro}
  (Eds.), \bibinfo{booktitle}{Concurrency Theory ({CONCUR}'09)}, vol.
  \bibinfo{volume}{5710} of \emph{\bibinfo{series}{Lect. Notes Comp. Sci.}},
  \bibinfo{publisher}{Springer}, \bibinfo{pages}{369--383},
  \bibinfo{year}{2009}.

\bibitem[{Hague and Lin(2011)}]{HL11}
\bibinfo{author}{M.~Hague}, \bibinfo{author}{A.~W. Lin}, \bibinfo{title}{Model
  Checking Recursive Programs with Numeric Data Types}, in:
  \bibinfo{editor}{G.~Gopalakrishnan}, \bibinfo{editor}{S.~Qadeer} (Eds.),
  \bibinfo{booktitle}{Computer Aided Verification ({CAV}'11)}, vol.
  \bibinfo{volume}{6806} of \emph{\bibinfo{series}{Lect. Notes Comp. Sci.}},
  \bibinfo{publisher}{Springer}, \bibinfo{pages}{743--759},
  \bibinfo{year}{2011}.

\bibitem[{David and Alla(1987)}]{DA87}
\bibinfo{author}{R.~David}, \bibinfo{author}{H.~Alla},
  \bibinfo{title}{Continuous {P}etri nets}, in: \bibinfo{booktitle}{Proceedings
  of the 8th European Workshop on Application and Theory of {P}etri nets},
  \bibinfo{pages}{275--294}, \bibinfo{year}{1987}.

\bibitem[{Fraca and Haddad(2015)}]{FH15}
\bibinfo{author}{E.~Fraca}, \bibinfo{author}{S.~Haddad},
  \bibinfo{title}{Complexity Analysis of Continuous Petri Nets},
  \bibinfo{journal}{Fundam. Inform.}
  \bibinfo{volume}{137}~(\bibinfo{number}{1}) (\bibinfo{year}{2015})
  \bibinfo{pages}{1--28}.

\bibitem[{P{\u a}un(1980)}]{Paun80}
\bibinfo{author}{G.~P{\u a}un}, \bibinfo{title}{A new generative device:
  valence grammars}, \bibinfo{journal}{Rev. Roumaine Math. Pures Appl.}
  \bibinfo{volume}{25}~(\bibinfo{number}{6}) (\bibinfo{year}{1980})
  \bibinfo{pages}{911--924}.

\bibitem[{Greibach(1978)}]{Grei78}
\bibinfo{author}{S.~A. Greibach}, \bibinfo{title}{Remarks on Blind and
  Partially Blind One-Way Multicounter Machines}, \bibinfo{journal}{Theor.
  Comput. Sci.} \bibinfo{volume}{7} (\bibinfo{year}{1978})
  \bibinfo{pages}{311--324}.

\bibitem[{Hoogeboom(2001)}]{Hoog01}
\bibinfo{author}{H.~J. Hoogeboom}, \bibinfo{title}{Context-Free Valence
  Grammars - Revisited}, in: \bibinfo{editor}{W.~Kuich},
  \bibinfo{editor}{G.~Rozenberg}, \bibinfo{editor}{A.~Salomaa} (Eds.),
  \bibinfo{booktitle}{Developments in Language Theory ({DLT}'01)}, vol.
  \bibinfo{volume}{2295} of \emph{\bibinfo{series}{Lect. Notes Comp. Sci.}},
  \bibinfo{publisher}{Springer}, \bibinfo{pages}{293--303},
  \bibinfo{year}{2001}.

\bibitem[{Fernau and Stiebe(2002)}]{FS02}
\bibinfo{author}{H.~Fernau}, \bibinfo{author}{R.~Stiebe},
  \bibinfo{title}{Sequential grammars and automata with valences},
  \bibinfo{journal}{Theor. Comput. Sci.}
  \bibinfo{volume}{276}~(\bibinfo{number}{1--2}) (\bibinfo{year}{2002})
  \bibinfo{pages}{377--405}.

\bibitem[{Buckheister and Zetzsche(2013)}]{BZ13}
\bibinfo{author}{P.~Buckheister}, \bibinfo{author}{G.~Zetzsche},
  \bibinfo{title}{Semilinearity and Context-Freeness of Languages Accepted by
  Valence Automata}, in:  \cite{DBLP:conf/mfcs/2013},
  \bibinfo{pages}{231--242}, \bibinfo{year}{2013}.

\bibitem[{Borosh and Treybing(1976)}]{BT76}
\bibinfo{author}{I.~Borosh}, \bibinfo{author}{L.~Treybing},
  \bibinfo{title}{Bounds on positive integral solutions of linear {D}iophantine
  equations}, \bibinfo{journal}{Proc.\ AMS} \bibinfo{volume}{55}
  (\bibinfo{year}{1976}) \bibinfo{pages}{299--304}.

\bibitem[{Gr{\"a}del(1989)}]{Grae89}
\bibinfo{author}{E.~Gr{\"a}del}, \bibinfo{title}{Dominoes and the complexity of
  subclasses of logical theories}, \bibinfo{journal}{Ann. Pure Appl. Logic}
  \bibinfo{volume}{43}~(\bibinfo{number}{1}) (\bibinfo{year}{1989})
  \bibinfo{pages}{1--30}.

\bibitem[{Haase(2014)}]{Haa14}
\bibinfo{author}{C.~Haase}, \bibinfo{title}{Subclasses of {Presburger}
  arithmetic and the weak {EXP} hierarchy}, in: \bibinfo{editor}{T.~A.
  Henzinger}, \bibinfo{editor}{D.~Miller} (Eds.), \bibinfo{booktitle}{Joint
  Meeting of Computer Science Logic {(CSL)} and Logic in Computer Science
  (LICS), {CSL-LICS}'14}, \bibinfo{publisher}{{ACM}},
  \bibinfo{pages}{47:1--47:10}, \bibinfo{year}{2014}.

\bibitem[{Ginsburg and Spanier(1966)}]{GS66}
\bibinfo{author}{S.~Ginsburg}, \bibinfo{author}{E.~Spanier},
  \bibinfo{title}{Semigroups, {P}resburger formulas and languages},
  \bibinfo{journal}{Pac. J. Math.} \bibinfo{volume}{16}~(\bibinfo{number}{2})
  (\bibinfo{year}{1966}) \bibinfo{pages}{285--296}.

\bibitem[{Garey and Johnson(1979)}]{GJ79}
\bibinfo{author}{M.~Garey}, \bibinfo{author}{D.~Johnson},
  \bibinfo{title}{Computers and Intractability: A Guide to the Theory of
  {NP}-Completeness}, \bibinfo{publisher}{W. H. Freeman \& Co.},
  \bibinfo{address}{New York, NY, USA}, \bibinfo{year}{1979}.

\bibitem[{Berman et~al.(2002)Berman, Karpinski, Larmore, Plandowski, and
  Rytter}]{BKLPR02}
\bibinfo{author}{P.~Berman}, \bibinfo{author}{M.~Karpinski},
  \bibinfo{author}{L.~L. Larmore}, \bibinfo{author}{W.~Plandowski},
  \bibinfo{author}{W.~Rytter}, \bibinfo{title}{On the Complexity of Pattern
  Matching for Highly Compressed Two-Dimensional Texts}, \bibinfo{journal}{J.
  Comput. Syst. Sci.} \bibinfo{volume}{65}~(\bibinfo{number}{2})
  (\bibinfo{year}{2002}) \bibinfo{pages}{332--350}.

\bibitem[{Chistikov and Majumdar(2014)}]{CM14}
\bibinfo{author}{D.~Chistikov}, \bibinfo{author}{R.~Majumdar},
  \bibinfo{title}{Unary Pushdown Automata and Straight-Line Programs}, in:
  \bibinfo{editor}{J.~Esparza}, \bibinfo{editor}{P.~Fraigniaud},
  \bibinfo{editor}{T.~Husfeldt}, \bibinfo{editor}{E.~Koutsoupias} (Eds.),
  \bibinfo{booktitle}{Automata, Languages, and Programming ({ICALP}'14), Part
  {II}}, vol. \bibinfo{volume}{8573} of \emph{\bibinfo{series}{Lect. Notes
  Comp. Sci.}}, \bibinfo{publisher}{Springer}, \bibinfo{pages}{146--157},
  \bibinfo{year}{2014}.

\bibitem[{Chistikov and Haase(2016)}]{CH16}
\bibinfo{author}{D.~Chistikov}, \bibinfo{author}{C.~Haase}, \bibinfo{title}{The
  Taming of the Semi-Linear Set}, in: \bibinfo{booktitle}{Automata, Languages
  and Programming ({ICALP}'16)}, \bibinfo{note}{to appear},
  \bibinfo{year}{2016}.

\bibitem[{Blondin et~al.(2016)Blondin, Finkel, Haase, and Haddad}]{BFHH15}
\bibinfo{author}{M.~Blondin}, \bibinfo{author}{A.~Finkel},
  \bibinfo{author}{C.~Haase}, \bibinfo{author}{S.~Haddad},
  \bibinfo{title}{Approaching the Coverability Problem Continuously}, in:
  \bibinfo{editor}{M.~Chechik}, \bibinfo{editor}{J.~Raskin} (Eds.),
  \bibinfo{booktitle}{Tools and Algorithms for the Construction and Analysis of
  Systems ({TACAS'16})}, vol. \bibinfo{volume}{9636} of
  \emph{\bibinfo{series}{Lect. Notes Comp. Sci.}},
  \bibinfo{publisher}{Springer}, \bibinfo{pages}{480--496},
  \bibinfo{year}{2016}.

\bibitem[{Chatterjee and Sgall(2013)}]{DBLP:conf/mfcs/2013}
\bibinfo{editor}{K.~Chatterjee}, \bibinfo{editor}{J.~Sgall} (Eds.),
  \bibinfo{title}{Mathematical Foundations of Computer Science 2013
  ({MFCS}'13)}, vol. \bibinfo{volume}{8087} of \emph{\bibinfo{series}{Lect.
  Notes Comp. Sci.}}, \bibinfo{publisher}{Springer}, \bibinfo{year}{2013}.

\end{thebibliography}


\end{document}